\documentclass{llncs}

\usepackage{xspace}
\usepackage{subfigure}
\usepackage{xcolor}
\usepackage{epsfig}
\usepackage{gensymb}
\usepackage[english]{babel}
\usepackage{graphicx}
\usepackage{amssymb}
\usepackage{multirow}
\usepackage{todonotes}
\usepackage{enumerate}
\usepackage{paralist}
\usepackage{amsmath}
\usepackage{cases}

\newcommand{\remove}[1]{}

\renewcommand{\int}{int}

\newtheorem{claimx}{Claim}

\newcommand{\cv}{collinear vertices}
\newcommand{\cfsl}{planar straight-line drawing}

\newcommand{\etal}{\emph{et al.}}

\renewenvironment{proof}
{{\bf Proof:}}{\hspace*{\fill}$\Box$\par\vspace{2mm}}

\newcommand{\rome}{$^{\dag}$} 
\newcommand{\ottawa}{$^{\diamond}$} 
\newcommand{\kit}{$^{\circ}$}

\begin{document}
\title{Drawing Planar Graphs with\\ Many Collinear Vertices 
\thanks{The research of Giordano Da Lozzo, Fabrizio Frati and Vincenzo Roselli was partially supported by MIUR Project ``AMANDA'' under PRIN 2012C4E3KT. The research of Vida Dujmovi\'c was partially supported by NSERC and Ontario Ministry of Research and Innovation.}
}
\author{Giordano Da Lozzo\rome, Vida Dujmovi\'c\ottawa, Fabrizio Frati\rome, \\Tamara Mchedlidze\kit, Vincenzo Roselli\rome}
\institute{
\rome University Roma Tre, Italy \\ 
\email{\{dalozzo,frati,roselli\}@dia.uniroma3.it}\\
\ottawa University of Ottawa, Canada\\
\email{vida.dujmovic@uottawa.ca}\\
\kit Karlsruhe Institute of Technology, Germany\\
\email{mched@iti.uka.de}
}
\maketitle
\pagestyle{plain}

\begin{abstract}
Consider the following problem: Given a planar graph $G$, what is the maximum 
number $p$ such that $G$ has a \cfsl\ with $p$ \cv? This problem resides at the core of several graph drawing problems, including  universal point subsets, untangling, and column planarity.  The following results are known for it: Every $n$-vertex planar graph has a \cfsl\ with  $\Omega(\sqrt{n})$ \cv; for every $n$, there is an $n$-vertex planar graph whose every \cfsl\ has $O(n^\sigma)$ \cv, where $\sigma<0.986$; every $n$-vertex planar graph of treewidth at most two has a \cfsl\ with $\Theta(n)$ \cv. We extend the linear bound to planar graphs of treewidth at most three and to triconnected cubic planar graphs. This (partially) answers two open problems posed by  Ravsky and Verbitsky [\emph{WG}~2011:295--306].
Similar results are not possible for all bounded treewidth planar graphs or for all bounded degree planar graphs. For planar graphs of treewidth at most three, our results also imply asymptotically tight bounds for all of the other above mentioned graph drawing problems.

\end{abstract}




\section{Introduction}\label{le:intro}


A subset $S$ of the vertices of a planar graph $G$ is a \emph{collinear set} if $G$ has a \cfsl\ where all the vertices in $S$ are collinear.  Ravsky and Verbitsky~\cite{DBLP:conf/wg/RavskyV11} consider the problem of determining the maximum cardinality of collinear sets in planar graphs. A stronger notion is defined as follows: a set $R\subseteq V(G)$  is a \emph{free collinear set} if a total order $<_R$ of $R$ exists such that, given any set of $|R|$ points on a line $\ell$, graph $G$ has a \cfsl\ where the vertices in $R$ are mapped to the given points and their order on $\ell$ matches the order $<_R$. Free collinear sets were first used 
(although not named) by Bose~\etal~\cite{DBLP:journals/dcg/BoseDHLMW09}; also, they were called \emph{free sets}  by Ravsky and Verbitsky~\cite{DBLP:conf/wg/RavskyV11}. Clearly, every free collinear set is also a  collinear set. In addition to this obvious relationship to collinear sets, free collinear sets have connections to other graph drawings problems, as will be discussed later in this introduction.




Based on the results in \cite{DBLP:journals/dcg/BoseDHLMW09}, Dujmovi\'c~\cite{DBLP:conf/gd/Dujmovic15} showed that every $n$-vertex planar graph has a free collinear set of cardinality at least $\sqrt{n/2}$. A natural question to consider would be whether a linear bound is possible for all planar graphs. Ravsky and Verbitsky~\cite{DBLP:conf/wg/RavskyV11} provided a negative answer to that question. In particular, they observed that if a planar triangulation has a large collinear set, then its dual graph has a cycle of proportional length. Since there are $m$-vertex triconnected cubic planar graphs whose longest cycle has length $O(m^\sigma)$~\cite{gw-sefg-73}, it follows that there are $n$-vertex planar graphs in which the cardinality of every collinear set is $O(n^\sigma)$. Here $\sigma$ is a known graph-theoretic constant called \emph{shortness exponent}, for which the best known upper bound is $\sigma<0.986$~\cite{gw-sefg-73}.

In addition to the natural open problem of closing the gap between the $\Omega(n^{0.5})$ and $O(n^\sigma)$ bounds for general $n$-vertex planar graphs, these results raise the question of which classes of planar graphs have (free) collinear sets of linear cardinality. Goaoc~\etal~\cite{DBLP:journals/dcg/GoaocKOSSW09} proved (implicitly) that $n$-vertex outerplanar graphs have free collinear sets of cardinality  $(n+1)/2$; this result was explicitly stated and proved by Ravsky and Verbitsky~\cite{DBLP:conf/wg/RavskyV11}. Ravsky and Verbitsky~\cite{DBLP:conf/wg/RavskyV11} also considerably strengthened that result by proving that all $n$-vertex planar graphs of treewidth at most two have free collinear sets of cardinality $n/30$; they also asked for other classes of graphs with (free) collinear sets of linear cardinality, calling special attention to planar graphs of bounded treewidth and to planar graphs of bounded degree. In this paper we prove the following results: 
\begin{enumerate}
\item every $n$-vertex planar graph of treewidth at most three has a free collinear set with cardinality $\lceil \frac{n-3}{8}\rceil$; 
\item every $n$-vertex triconnected cubic planar graph has a collinear set with cardinality $\lceil \frac{n}{4}\rceil$; and 
\item every planar graph of treewidth $k$ has a collinear set with cardinality $\Omega(k^2)$.
\end{enumerate}

Our first result generalizes the previous result on planar graphs of treewidth at most two~\cite{DBLP:conf/wg/RavskyV11}. As noted by Ravsky and Verbitsky in the full version of their paper \cite[Corollary 3.5]{DBLP:journals/corr/abs-0806-0253}, there are $n$-vertex planar graphs of treewidth at most  $8$ whose largest collinear set has cardinality $o(n)$. To obtain that, the authors show a construction relying on the dual of the Barnette-Bos\'ak-Lederberg's non-Hamiltonian cubic triconnected planar graph. It can be shown that the dual of Tutte's graph has treewidth $5$, thus if one relies on that dual instead, the sub-linear upper bound holds true for planar graphs of treewidth at most $5$. Thus, our first result leaves $k=4$ as the only remaining open case for the question of whether planar graphs of treewidth at most $k$ admit (free) collinear sets with linear cardinality.

Our second result provides the first linear lower bound on the cardinality of collinear sets for a fairly wide class of bounded-degree planar graphs. The result cannot be extended to all bounded-degree planar graphs. In particular it cannot be extended to planar graphs of degree at most $7$, since there exist $n$-vertex planar triangulations of maximum degree $7$ whose dual graph has a longest cycle of length $o(n)$~\cite{DBLP:journals/dm/Owens81}.

Finally, our third result improves the $\Omega(\sqrt{n})$ bound on the cardinality of collinear sets in general planar graphs for all planar graphs whose treewidth is~$\omega(\sqrt[4]{n})$. 

We now discuss applications of our results to other graph drawing problems. Since our first result gives {\em free} collinear sets, its consequences are broader. 

A {\em column planar set} in a planar graph $G$ is a set $Q\subseteq V(G)$ satisfying the following property: there exists a function $\gamma: Q \rightarrow \mathbb{R}$ such that, for any function $\lambda: Q \rightarrow \mathbb{R}$, there exists a \cfsl\ of $G$ in which each vertex $v\in Q$ is mapped to point $(\gamma(v),\lambda(v))$. Column planar sets were defined by Evans~\etal~\cite{DBLP:conf/gd/EvansKSS14} motivated by applications to partial simultaneous geometric embeddings\footnote{The original definition by Evans~\etal~\cite{DBLP:conf/gd/EvansKSS14} had also an extra condition that required the point set composed of the points $(\gamma(v),\lambda(v))$ for all $v\in Q$ not to have three points on a line.}. They proved that $n$-vertex trees have column planar sets of size $14n/17$. The lower bounds in all our three results carry over to the size of column planar sets for the corresponding graph classes.

A \emph{universal point subset} for the $n$-vertex planar graphs is a set $P$ of $k\leq n$ points in the plane such that, for every $n$-vertex planar graph $G$, there exists a \cfsl\ of $G$ in which $k$ vertices are placed at the $k$ points in $P$. Universal point subsets were introduced by Angelini~\etal~\cite{DBLP:conf/isaac/AngeliniBEHLMMO12}.  Every set of $n$ points in general position is a universal point subset for the $n$-vertex outerplanar graphs \cite{GMPP,DBLP:journals/comgeo/Bose02,DBLP:conf/cccg/CastanedaU96} and every set of $\sqrt{n/2}$ points in the plane is a universal point subset for the $n$-vertex planar graphs \cite{DBLP:conf/gd/Dujmovic15}. As a corollary of our first result, we obtain that every set of $\lceil \frac{n-3}{8}\rceil$ points in the plane is a universal point subset for the $n$-vertex planar graphs of treewidth at most three.

Given a straight-line drawing of a planar graph, possibly with crossings, to \emph{untangle} it means to assign new locations to some of its vertices so that the resulting straight-line drawing is planar. The goal is to do so while \emph{keeping fixed} (i.e., not changing the location of) as many vertices as possible. Several papers have studied the untangling problem~\cite{DBLP:journals/dcg/PachT02,DBLP:journals/siamdm/CanoTU14,DBLP:journals/dcg/Cibulka10,DBLP:journals/dcg/BoseDHLMW09,DBLP:journals/dcg/GoaocKOSSW09,DBLP:journals/dam/KangPRSV11,DBLP:conf/wg/RavskyV11}. It is known that general $n$-vertex planar graphs can be untangled while keeping $\Omega(n^{0.25})$ vertices fixed \cite{DBLP:journals/dcg/BoseDHLMW09} and that there are $n$-vertex planar graphs that cannot be untangled while keeping $\Omega(n^{0.4948})$ vertices fixed \cite{DBLP:journals/siamdm/CanoTU14}. Asymptotically tight bounds are known for paths \cite{DBLP:journals/dcg/Cibulka10}, trees \cite{DBLP:journals/dcg/GoaocKOSSW09}, outerplanar graphs \cite{DBLP:journals/dcg/GoaocKOSSW09}, and planar graphs of treewidth two \cite{DBLP:conf/wg/RavskyV11}. As a corollary of our first result, we obtain that every $n$-vertex planar graph of treewidth at most three can be untangled while keeping $\Omega(\sqrt{n})$ vertices fixed. This bound is the best possible, as there are forests of stars that cannot be untangled while keeping $\omega(\sqrt{n})$ vertices fixed \cite{DBLP:journals/dcg/BoseDHLMW09}. Our result generalizes previous results on trees, outerplanar graphs and planar graphs of treewidth at most two.

\section{Preliminaries} \label{le:preliminaries}

The graphs called \emph{$k$-trees} are defined recursively as follows. A complete graph on $k+1$ vertices is a $k$-tree. If $G$ is a $k$-tree, then the graph obtained by adding a new vertex to $G$ and making it adjacent to all the vertices in a $k$-clique of $G$ is a $k$-tree. The {\em treewidth} of a graph $G$ is the minimum $k$ such that $G$ is a subgraph of some $k$-tree. 

A connected {\em plane graph} $G$ is a connected planar graph together with a {\em plane embedding}, that is, an equivalence class of planar drawings of $G$, where two planar drawings are {\em equivalent} if they have the same {\em rotation system} (i.e., the same clockwise order of the edges incident to each vertex) and the same {\em outer face} (i.e., the unbounded face is delimited by the same walk). We always think about a plane graph $G$ as if it is drawn according to its plane embedding; also, when we talk about a planar drawing of $G$, we always mean that it respects the plane embedding of $G$. The {\em interior} of $G$ is the closure of the union of the internal faces of $G$. We associate with a subgraph $H$ of $G$ the plane embedding obtained from the one of $G$ by deleting vertices and edges not in $H$. 

We denote the degree of a vertex $v$ in a graph $G$ by $\delta_G(v)$. A graph is {\em cubic} ({\em subcubic}) if every vertex has degree $3$ (resp.\ at most $3$). Let $G$ be a graph and $U\subseteq V(G)$. We denote by $G-U$ the graph obtained from $G$ by removing the vertices in $U$ and their incident edges. The subgraph of $G$ {\em induced} by $U$ has $U$ as vertex set and has an edge $e\in E(G)$ if and only if both its end-vertices are in $U$. Let $H$ be a subgraph of $G$; then $H$ is {\em induced} if $H$ is induced by $V(H)$. If $v\in V(G)-V(H)$, we denote by $H\cup \{v\}$ the subgraph of $G$ composed of $H$ and of the isolated vertex $v$. An {\em $H$-bridge} $B$ is either a {\em trivial $H$-bridge} -- an edge of $G$ not in $H$ with both end-vertices in $H$ -- or a non-trivial $H$-bridge -- a connected component of $G-V(H)$ together with the edges from that component to $H$. The vertices in $V(H)\cap V(B)$ are called {\em attachments}.


Let $G$ be a connected graph. A {\em cut-vertex} is a vertex whose removal disconnects $G$. If $G$ has no cut-vertex and it is not a single edge, then it is {\em biconnected}. A {\em biconnected component} of $G$ is a maximal (with respect to both vertices and edges) biconnected subgraph of $G$. Let $G$ be a biconnected graph. A {\em separation pair} is a pair of vertices whose removal disconnects $G$. If $G$ has no separation pair, then it is {\em triconnected}. Given a separation pair $\{a,b\}$ in a biconnected graph $G$, an {\em $\{a,b\}$-component} is either a {\em trivial $\{a,b\}$-component} -- edge $(a,b)$ -- or a {\em non-trivial $\{a,b\}$-component} -- a subgraph of $G$ induced by $a$, $b$, and the vertices of a connected component of $G-\{a,b\}$. 



\section{From a Geometric to a Topological Problem} \label{le:topology}

In this section we show that the problem of determining a large collinear set in a planar graph, which is geometric by definition, can be transformed into a purely topological problem. This result may be useful to obtain bounds for the size of collinear sets in classes of planar graphs different from the ones we studied in this paper.

Given a planar drawing $\Gamma$ of a plane graph $G$, we say that an open simple (i.e., non-self-intersecting) curve $\lambda$ is {\em good} for $\Gamma$ if, for each edge $e$ of $G$, curve $\lambda$ either entirely contains $e$ or has at most one point in common with $e$ (if $\lambda$ passes through an end-vertex of $e$, that counts as a common point). Clearly, the existence of a good curve passing through a certain sequence of vertices, edges, and faces of $G$ does not depend on the actual drawing $\Gamma$, but only on the plane embedding of $G$. For this reason we often talk about the existence of good curves in plane graphs, rather than in their planar drawings. We denote by $R_{G,\lambda}$ the only unbounded region of the plane defined by $G$ and $\lambda$. Curve $\lambda$ is {\em proper} if both its end-points are incident to $R_{G,\lambda}$. We have the following.

\begin{theorem} \label{th:topology}
A plane graph $G$ has a \cfsl\ with $x$ collinear vertices if and only if $G$ has a proper good curve that passes through $x$ vertices of $G$.
\end{theorem}

\begin{proof}
For the necessity, assume that $G$ has a \cfsl\ $\Gamma$ with $x$ vertices lying on a common line $\ell$. We transform $\ell$ into a straight-line segment $\lambda$ by cutting off two disjoint half-lines of $\ell$ in the outer face of $G$. This immediately implies that $\lambda$ is proper. Further, $\lambda$ passes through $x$ vertices of $G$ since $\ell$ does. Finally, if an edge $e$ has two common points with $\lambda$, then $\lambda$ entirely contains it, since $\lambda$ is a straight-line segment and since $e$ is a straight-line segment in $\Gamma$.





For the sufficiency, assume that $G$ has a proper good curve $\lambda$ passing through $x$ of its vertices; see Fig.~\ref{fig:characterization}(a). Augment $G$ by adding to it (refer to Fig.~\ref{fig:characterization}(b)): (i) a dummy vertex at each proper crossing between an edge and $\lambda$; (ii) two dummy vertices at the end-points $a$ and $b$ of $\lambda$; (iii) an edge between any two consecutive vertices of $G$ along $\lambda$, which now represents a path $(a,\dots,b)$ of $G$; (iv) two dummy vertices $d_1$ and $d_2$ in $R_{G,\lambda}$; and (v) edges in $R_{G,\lambda}$ connecting each of $d_1$ and $d_2$ with each of $a$ and $b$ so that cycles $C_1=(d_1,a,\dots,b)$ and $C_2=(d_2,a,\dots,b)$ are embedded in this counter-clockwise and clockwise direction in $G$, respectively. For $i=1,2$, let $G_i$ be the subgraph of $G$ induced by the vertices of $C_i$ or inside it. Triangulate the internal faces of $G_i$ with dummy vertices and edges, so that there are no edges between non-consecutive vertices of $C_i$; indeed, these edges do not exist in the original graph $G$, given that $\lambda$ is good.

\begin{figure}[htb]
\begin{center}
\begin{tabular}{c c c c}
\mbox{\includegraphics[scale=.5]{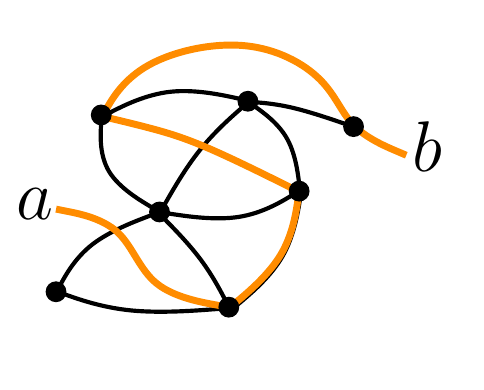}} \hspace{1mm} &
\mbox{\includegraphics[scale=.5]{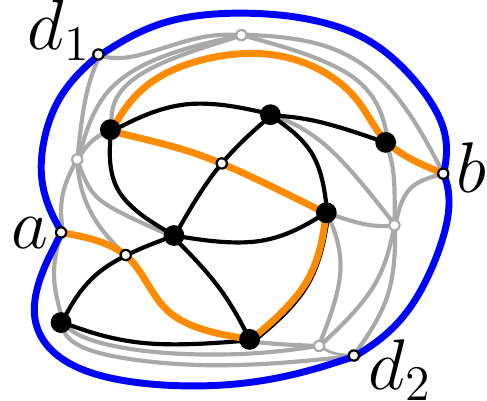}} \hspace{1mm} &
\mbox{\includegraphics[scale=.5]{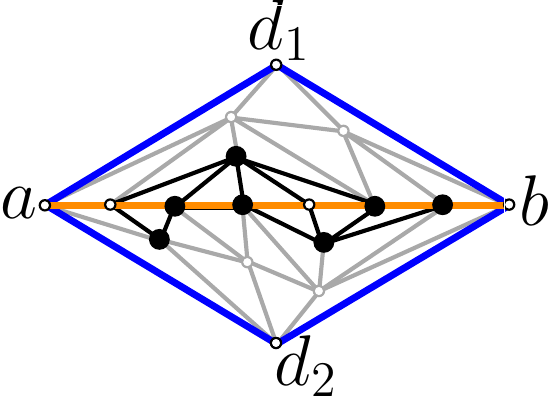}} \hspace{1mm} &
\mbox{\includegraphics[scale=.5]{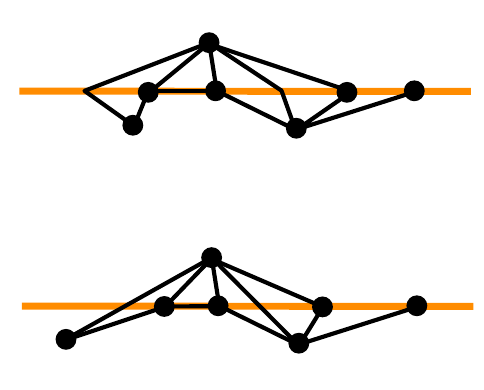}}\\
(a) \hspace{1mm} & (b) \hspace{1mm} & (c) \hspace{1mm} & (d)
\end{tabular}
\caption{(a) A proper good curve $\lambda$ (orange) for a plane graph $G$ (black). (b) Augmentation of $G$ with dummy vertices and edges. (c) A \cfsl\ of the augmented graph $G$. (d) Planar polyline (top) and straight-line (bottom) drawings of the original graph $G$.}
\label{fig:characterization}
\end{center}
\end{figure}

Represent $C_1$ as a convex polygon $Q_1$ whose all vertices, except for $d_1$, lie along a horizontal line $\ell$, with $a$ to the left of $b$ and $d_1$ above $\ell$; see Fig.~\ref{fig:characterization}(c). Graph $G_1$ is triconnected, as it contains no edge between any two non-consecutive vertices of its only non-triangular face. Thus, a \cfsl\ of $G_1$ in which $C_1$ is represented by $Q_1$ exists~\cite{t-hdg-63}. Analogously, represent $C_2$ as a convex polygon $Q_2$ whose all vertices, except for $d_2$, lie at the same points as in $Q_1$, with $d_2$ below $\ell$. Construct a \cfsl\ of $G_2$ in which $C_2$ is represented by $Q_2$.

Removing the dummy vertices and edges results in a planar drawing $\Gamma$ of the original graph $G$ in which each edge $e$ is a $y$-monotone curve; see Fig.~\ref{fig:characterization}(d). In particular, the fact that $\lambda$ crosses at most once $e$ ensures that $e$ is either a straight-line segment or is composed of two straight-line segments that are one below and one above $\ell$ and that share an end-point on $\ell$. A \cfsl\ $\Gamma'$ of $G$ in which the $y$-coordinate of each vertex is the same as in $\Gamma$ always exists, as proved in~\cite{efln-sldahgc-06,pt-mdpg-04}. Since $\lambda$ passes through $x$ vertices of $G$, we have that $x$ vertices of $G$ lie along $\ell$ in $\Gamma'$. 
\end{proof}

Theorem~\ref{th:topology} can be stated for planar graphs without a given plane embedding as follows: A planar graph has a collinear set with $x$ vertices if and only if it admits a plane embedding for which a proper good curve can be drawn that passes through $x$ of its vertices. While this version of Theorem~\ref{th:topology} might be more general, it is less useful for us, so we preferred to explicitly state its version for plane graphs.



\section{Planar Graphs with Treewidth at most Three} \label{le:3-trees}

In this section we prove the following theorem.

\begin{theorem} \label{th:3-trees}
Every $n$-vertex plane graph of treewidth at most three admits a \cfsl\ with at least $\lceil \frac{n-3}{8}\rceil$ collinear vertices.
\end{theorem}


For technical reasons, we regard a plane cycle with three vertices as a plane $3$-tree. Then every plane graph $G$ with $n\geq 3$ vertices and treewidth at most three can be augmented with dummy edges to a plane $3$-tree $G'$~\cite{kv-npp3t-12} which is a plane triangulation. A \cfsl\ of $G$ with $\lceil \frac{n-3}{8}\rceil$ collinear vertices can be obtained from a \cfsl\ of $G'$ with $\lceil \frac{n-3}{8}\rceil$ collinear vertices by removing the inserted dummy edges. Thus for the remainder of this section, we assume that $G$ is a plane $3$-tree. 

By Theorem~\ref{th:topology} it suffices to prove that $G$ admits a proper good curve passing through $\lceil \frac{n-3}{8}\rceil$ vertices of $G$. Let $u$, $v$, and $z$ be the external vertices of $G$. If $n=3$, then $G$ does not contain any internal vertex and we say that it is {\em empty}. If $G$ is not empty, let $w$ be the unique internal vertex of $G$ adjacent to all of $u$, $v$, and $z$; we say that $w$ is the {\em central vertex} of $G$. Let $G_1$, $G_2$, and $G_3$ be the plane $3$-trees which are the subgraphs of $G$ whose outer faces are delimited by cycles $(u,v,w)$, $(u,z,w)$, and $(v,z,w)$. We call $G_1$, $G_2$, and $G_3$ {\em children} of $G$ and {\em children} of $w$.    

We associate to each internal vertex $x$ of $G$ a plane $3$-tree $G(x)$, which is a subgraph of $G$, as follows. We associate $G$ to $w$ and we recursively associate plane $3$-trees to the internal vertices of the children $G_1$, $G_2$, and $G_3$ of $G$. Note that $x$ is the central vertex of the plane $3$-tree $G(x)$ associated to it.

We now introduce a classification of the internal vertices of $G$; see Fig.~\ref{fig:plane3-tree-intro}(a). Consider an internal vertex $x$ of $G$. We say that $x$ is of {\em type A}, {\em B}, {\em C}, or {\em D} if, respectively, $3$, $2$, $1$, or $0$ of the children of $G(x)$ are empty. We denote by $a(G)$, $b(G)$, $c(G)$, and $d(G)$ the number of internal vertices of $G$ of type A, B, C, and D, respectively. Let $m=n-3$ be the number of internal vertices of $G$. 

\begin{figure}[htb]
\begin{center}
\begin{tabular}{c c c}
\mbox{\includegraphics[height=.25\textwidth]{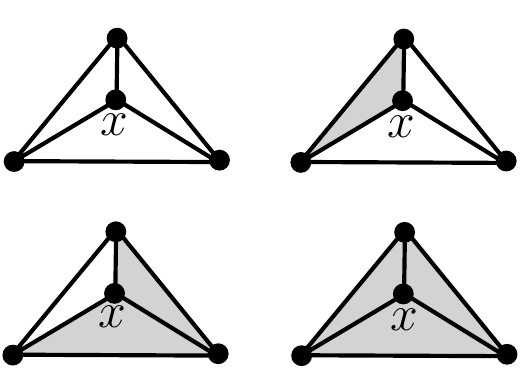}} \hspace{3mm} &
\mbox{\includegraphics[height=.25\textwidth]{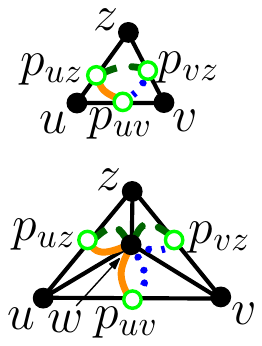}} \hspace{3mm} &
\mbox{\includegraphics[height=.25\textwidth]{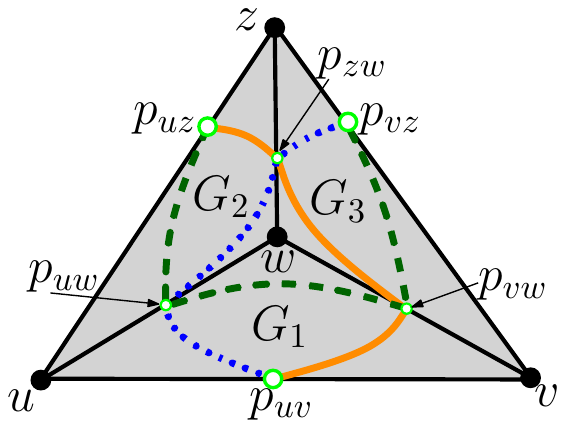}}\\
(a) \hspace{3mm} & (b) \hspace{3mm} & (c)
\end{tabular}
\caption{(a) A vertex $x$ of type A (top-left), B (top-right), C (bottom-left), and D (bottom-right). (b) Curves $\lambda_{u}(G)$ (solid), $\lambda_{v}(G)$  (dotted), and $\lambda_{z}(G)$  (dashed) if $m=0$ (top) and $m=1$ (bottom). (c) Curves $\lambda_{u}(G)$, $\lambda_{v}(G)$, and $\lambda_{z}(G)$ if $w$ is of type C or D.}
\label{fig:plane3-tree-intro}
\end{center}
\end{figure}

In the following we present an algorithm that computes three proper good curves $\lambda_{u}(G)$, $\lambda_v(G)$, and $\lambda_z(G)$ lying in the interior of $G$. For every edge $(x,y)$ of $G$, let $p_{xy}$ be an arbitrary internal point of $(x,y)$. The end-points of $\lambda_u(G)$ are $p_{uv}$ and $p_{uz}$, the end-points of $\lambda_v(G)$ are $p_{uv}$ and $p_{vz}$, and the end-points of $\lambda_z(G)$ are $p_{uz}$ and $p_{vz}$. Although each of $\lambda_{u}(G)$, $\lambda_v(G)$, and $\lambda_z(G)$ is a good curve, any two of these curves might cross each other arbitrarily and might pass through the same vertices of $G$. Each of these curves passes through all the internal vertices of $G$ of type A, through no vertex of type C or D, and through ``some'' vertices of type B. We will prove that the total number of internal vertices of $G$ curves $\lambda_{u}(G)$, $\lambda_v(G)$, and $\lambda_z(G)$ pass through is at least $\frac{3m}{8}$, hence one of them passes through at least $\lceil \frac{m}{8} \rceil$ internal vertices of~$G$. 

Curves $\lambda_{u}(G)$, $\lambda_v(G)$, and $\lambda_z(G)$ are constructed by induction on $m$. In the base case we have $m\leq 1$; refer to Fig.~\ref{fig:plane3-tree-intro}(b). If $m=0$, then $\lambda_{u}(G)$ starts at $p_{uv}$, traverses the internal face $(u,v,z)$ of $G$, and ends at $p_{uz}$. Curves $\lambda_{v}(G)$ and $\lambda_{z}(G)$ are defined analogously. If $m=1$, then $\lambda_{u}(G)$ starts at $p_{uv}$, traverses the internal face $(u,v,w)$ of $G$, passes through the central vertex $w$ of $G$, traverses the internal face $(u,z,w)$ of $G$, and ends at $p_{uz}$. Curves $\lambda_{v}(G)$ and $\lambda_{z}(G)$ are defined analogously. 


If $m>1$, then the central vertex $w$ of $G$ is of one of types B--D. If $w$ is of type C or D, then proper good curves are inductively constructed for the children of $G$ and composed to obtain $\lambda_{u}(G)$, $\lambda_v(G)$, and $\lambda_z(G)$. If $w$ is of type B, then a maximal sequence of vertices of type B starting at $w_1=w$ is considered; this sequence is called a {\em B-chain}. While the only child $H_i$ of the last vertex $w_i$ in the sequence has a central vertex $w_{i+1}$ of type B, the sequence is enriched with $w_{i+1}$; once $w_{i+1}$ is not of type B, induction is applied on $H_{i}$, and the three curves obtained by induction are composed with curves passing through vertices of the B-chain to get $\lambda_{u}(G)$, $\lambda_v(G)$, and $\lambda_z(G)$.


Assume first that $w$ is of type C or D. Refer to Fig.~\ref{fig:plane3-tree-intro}(c). Inductively construct curves $\lambda_{u}(G_1)$, $\lambda_v(G_1)$, and $\lambda_w(G_1)$ for $G_1$, curves $\lambda_{u}(G_2)$, $\lambda_z(G_2)$, and $\lambda_w(G_2)$ for $G_2$, and curves $\lambda_{v}(G_3)$, $\lambda_z(G_3)$, and $\lambda_w(G_3)$ for $G_3$. Let 
\begin{eqnarray*}
\lambda_{u}(G)  &=&  \lambda_v(G_1)\cup \lambda_w(G_3)\cup \lambda_z(G_2), \\
\lambda_{v}(G)  &=&  \lambda_u(G_1)\cup \lambda_w(G_2)\cup \lambda_z(G_3), \textrm{ and } \\ 
\lambda_{z}(G)  &=&  \lambda_u(G_2)\cup \lambda_w(G_1)\cup \lambda_v(G_3).
\end{eqnarray*}

Next, consider the case in which $w$ is of type B. In order to describe how to construct curves $\lambda_{u}(G)$, $\lambda_v(G)$, and $\lambda_z(G)$, we need to further explore the structure of $G$. 

Let $H_0=G$, let $w_1=w$, and let $H_1$ be the only non-empty child of $G$. 
We define three paths $P_u$, $P_v$, and $P_z$ as described in Table~\ref{tab:h1-paths}, depending on which among $u$, $v$, $z$, and $w$ are the external vertices of $H_1$.

\begin{table}\centering
	\vspace{-1.5em}
\begin{tabular}{c|c|c|c}
	external vertices of $H_1$	& $P_u$ 		 & $P_v$ 	& $P_z$\\\hline
	$v,w,z$						& $(u,w)$		& $(v)$ 	& $(z)$ \\\hline
	$u,w,z$						& $(u)$			 & $(v,w)$	& $(z)$ \\\hline
	$u,v,w$						& $(u)$			 & $(v)$	  & $(z,w)$\vspace{1.4pt}
\end{tabular}
\caption{Definition of $P_u$, $P_v$, and $P_z$ depending on the external vertices of $H_1$.}
\label{tab:h1-paths}
\vspace{-2em}
\end{table}

Now suppose that, for some $i\geq 1$, a sequence $w_1,\dots,w_i$ of vertices of type B, a sequence $H_0,H_1,\dots,H_i$ of plane $3$-trees, and three paths $P_u$, $P_v$, and $P_z$ (possibly single vertices or edges) have been defined so that the following properties hold true: 
\begin{enumerate}[$(1)$] 
\item for $1\leq j\leq i$, vertex $w_j$ is the central vertex of $H_{j-1}$ and $H_j$ is the only non-empty child of $H_{j-1}$; 
\item $P_u$, $P_v$, and $P_z$ are vertex-disjoint and each of them is induced in $G$; and 
\item $P_u$, $P_v$, and $P_z$ connect $u$, $v$, and $z$ with the three external vertices $u'$, $v'$, and $z'$ of $H_i$, respectively. 
\end{enumerate}

Properties (1)--(3) are indeed satisfied with $i=1$. Consider the central vertex of $H_i$ and denote it by $w_{i+1}$. 

If $w_{i+1}$ is of type B, then let $H_{i+1}$ be the only non-empty child of $H_{i}$. If cycle $(v',z',w_{i+1})$ delimits the outer face of $H_{i+1}$, add edge $(u',w_{i+1})$ to $P_u$ and leave $P_v$ and $P_z$ unaltered. The cases in which cycles $(u',z',w_{i+1})$ or $(u',v',w_{i+1})$ delimit the outer face of $H_{i+1}$ can be dealt with analogously. Properties (1)--(3) are clearly satisfied by the described construction. 

If $w_{i+1}$ is not of type B, we call the sequence $w_1,\dots,w_i$ a {\em B-chain} of $G$; note that all of $w_1,\dots,w_i$ are of type B. For simplicity of notation,  let $H=H_i$. We denote the vertices of $P_u$, $P_v$, and $P_z$ as follows: $P_u$$=$$(u=u_1,u_2,\dots,u_U=u')$, $P_v$$=$$(v=v_1,v_2,\dots,v_V=v')$, and $P_z$$=$$(z=z_1,z_2,\dots,z_Z=z')$; also, define cycles $C_{uv}$$=$$P_u \cup (u,v) \cup P_v \cup (u',v')$, $C_{uz}$$=$$P_u \cup (u,z) \cup P_z \cup (u',z')$, and $C_{vz}$$=$$P_v \cup (v,z) \cup P_z \cup (v',z')$. Each of these cycles contains no vertex in its interior; also, every edge in the interior of $C_{uv}$, $C_{uz}$, or $C_{vz}$ connects two vertices on distinct paths among $P_u$, $P_v$, and $P_z$, given that each of these paths is induced. We are going to use the following (a similar lemma can be stated for $C_{uz}$ and $C_{vz}$).




%

\begin{lemma} \label{le:outerplanar-draw}
Let $p_1$ and $p_2$ be two points on the boundary of $C_{uv}$, possibly coinciding with vertices of $C_{uv}$, and not both on the same edge of $G$. There exists a good curve connecting $p_1$ and $p_2$, lying inside $C_{uv}$, except at its end-points, and intersecting every edge of $G$ inside $C_{uv}$ at most once.
\end{lemma} 

\begin{proof}
The lemma has a simple geometric proof. Represent $C_{uv}$ as a strictly-convex polygon and draw the edges of $G$ inside $C_{uv}$ as straight-line segments. Then the straight-line segment $\overline{p_1 p_2}$ is a good curve satisfying the requirements of the lemma.
\end{proof}


We now describe how to construct curves $\lambda_{u}(G)$, $\lambda_v(G)$, and $\lambda_z(G)$. First, inductively construct curves $\lambda_{u'}(H)$, $\lambda_{v'}(H)$, and $\lambda_{z'}(H)$ for $H$. The construction of $\lambda_{u}(G)$, $\lambda_v(G)$, and $\lambda_z(G)$ varies based on how many among $P_u$, $P_v$, and $P_z$ are single vertices. Observe that not all of $P_u$, $P_v$, and $P_z$ are single vertices, as $w_1\neq u,v,z$. 

Suppose first that none of $P_u$, $P_v$, and $P_z$ is a single vertex, as in Fig.~\ref{fig:plane3-tree-construction}(a). We describe how to construct $\lambda_{u}(G)$, as the construction of $\lambda_{v}(G)$ and $\lambda_{z}(G)$ is analogous. 

\begin{figure}[htb]
\begin{center}
\begin{tabular}{c c c}
\mbox{\includegraphics[scale=.6]{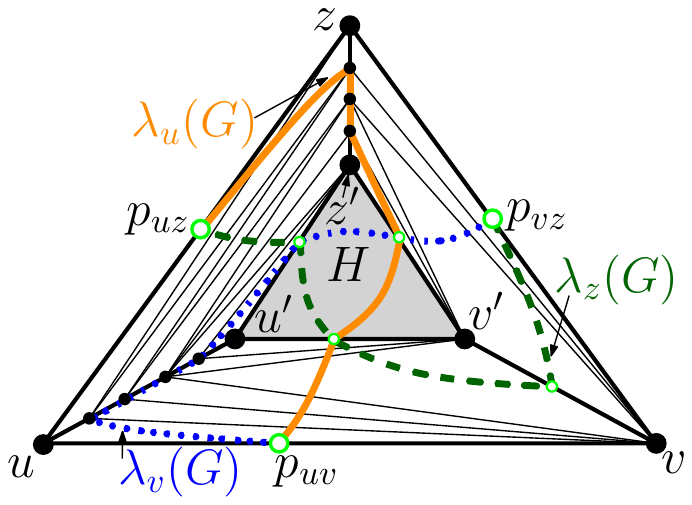}} \hspace{1mm} &
\mbox{\includegraphics[scale=.6]{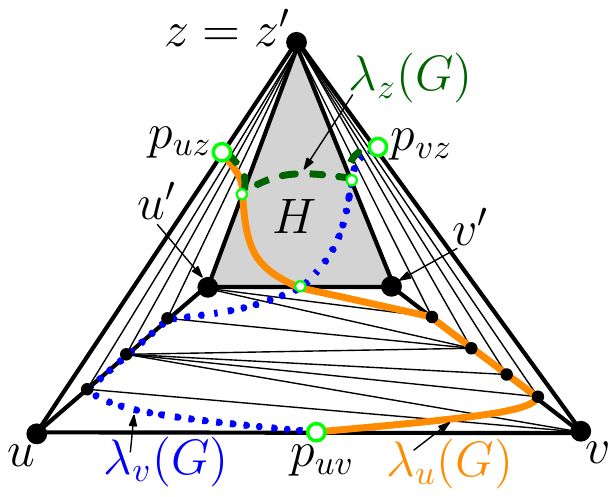}} \hspace{1mm} &
\mbox{\includegraphics[scale=.6]{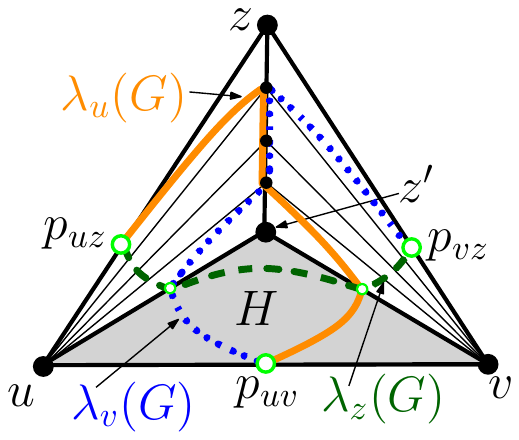}}\\
(a) \hspace{1mm} & (b) \hspace{1mm} & (c)
\end{tabular}
\caption{Construction of $\lambda_{u}(G)$, $\lambda_{v}(G)$, and $\lambda_{z}(G)$ if $w$ is of type B. (a) None of $P_u$, $P_v$, and $P_z$ is a single vertex. (b) $P_z$ is a single vertex while $P_u$ and $P_v$ are not. (c) $P_u$ and $P_v$ are single vertices while $P_z$ is not.}
\label{fig:plane3-tree-construction}
\end{center}
\end{figure}

\begin{itemize}
\item If $Z>2$, then $\lambda_{u}(G)$ consists of curves $\lambda^0_{u},\dots,\lambda^4_{u}$. Curve $\lambda^0_{u}$ lies inside $C_{uz}$ and connects $p_{uz}$ with $z_2$, which is internal to $P_z$ since $Z>2$; curve $\lambda^1_{u}$ coincides with path $(z_2,\dots,z_{Z-1})$ (the path consists of a single vertex if $Z=3$); curve $\lambda^2_{u}$ lies inside $C_{vz}$ and connects $z_{Z-1}$ with $p_{v'z'}$; curve $\lambda^3_{u}$ coincides with $\lambda_{v'}(H)$; finally, $\lambda^4_{u}$ lies inside $C_{uv}$ and connects $p_{u'v'}$ with $p_{uv}$. Curves $\lambda^0_{u}$, $\lambda^2_{u}$, and $\lambda^4_{u}$ are constructed as in Lemma~\ref{le:outerplanar-draw}. 
\item If $Z=2$, then $\lambda_{u}$ consists of curves $\lambda^1_{u},\dots,\lambda^4_{u}$. Curve $\lambda^1_{u}$ lies inside $C_{uz}$ and connects $p_{uz}$ with $p_{zz'}$; curve $\lambda^2_{u}$ lies inside $C_{vz}$ and connects $p_{zz'}$ with $p_{v'z'}$; curves $\lambda^3_{u}$ and $\lambda^4_{u}$ are defined as in the case $Z>2$. Curves $\lambda^1_{u}$, $\lambda^2_{u}$, and $\lambda^4_{u}$ are constructed as in Lemma~\ref{le:outerplanar-draw}.
\end{itemize}

Suppose next that one of $P_u$, $P_v$, and $P_z$, say $P_z$, is a single vertex, as in Fig.~\ref{fig:plane3-tree-construction}(b). We describe how to construct $\lambda_{u}(G)$ and $\lambda_{z}(G)$; the construction of $\lambda_{v}(G)$ is analogous to the one of $\lambda_{u}(G)$. Curve $\lambda_{z}(G)$ consists of curves $\lambda^0_{z},\lambda^1_{z},\lambda^2_{z}$. Curve $\lambda^0_{z}$ lies inside $C_{uz}$ and connects $p_{uz}$ with $p_{u'z}$; curve $\lambda^1_{z}$ coincides with $\lambda_{z'}(H)$; curve $\lambda^2_{z}$ lies inside $C_{vz}$ and connects $p_{v'z}$ with $p_{vz}$. Curves $\lambda^0_{z}$ and $\lambda^2_{z}$ are constructed as in Lemma~\ref{le:outerplanar-draw}. Curve $\lambda_{u}(G)$ is constructed as follows.

\begin{itemize}
\item If $V>2$, then $\lambda_{u}(G)$ consists of curves $\lambda^0_{u},\dots,\lambda^4_{u}$. Curve $\lambda^0_{u}$ lies inside $C_{uv}$ and connects $p_{uv}$ with $v_2$, which is internal to $P_v$ since $V>2$; curve $\lambda^1_{u}$ coincides with path $(v_2,\dots,v_{V-1})$ (the path consists of a single vertex if $V=3$); curve $\lambda^2_{u}$ lies inside $C_{uv}$ and connects $v_{V-1}$ with $p_{u'v'}$; curve $\lambda^3_{u}$ coincides with $\lambda_{u'}(H)$; finally, $\lambda^4_{u}$ coincides with $\lambda^0_{z}$. Curves $\lambda^0_{u}$, $\lambda^2_{u}$, and $\lambda^4_{u}$ are constructed as in Lemma~\ref{le:outerplanar-draw}.
\item If $V=2$, then $\lambda_{u}(G)$ consists of curves $\lambda^0_{u},\lambda^1_{u},\lambda^2_{u}$. Curve $\lambda^0_{u}$ lies inside $C_{uv}$ and connects $p_{uv}$ with $p_{u'v'}$; curve $\lambda^1_{u}$ coincides with $\lambda_{u'}(H)$; curve $\lambda^2_{u}$ coincides with $\lambda^0_{z}$. Curves $\lambda^0_{u}$ and $\lambda^2_{u}$ are constructed as in Lemma~\ref{le:outerplanar-draw}.
\end{itemize}

Suppose finally that two of $P_u$, $P_v$, and $P_z$, say $P_u$ and $P_v$, are single vertices, as in Fig.~\ref{fig:plane3-tree-construction}(c). We describe how to construct $\lambda_{u}(G)$ and $\lambda_{z}(G)$; the construction of $\lambda_{v}(G)$ is analogous to the one of $\lambda_{u}(G)$. Curve $\lambda_{z}(G)$ consists of curves $\lambda^0_{z},\lambda^1_{z},\lambda^2_{z}$. Curve $\lambda^0_{z}$ lies inside $C_{uz}$ and connects $p_{uz}$ with $p_{uz'}$; curve $\lambda^1_{z}$ coincides with $\lambda_{z'}(H)$; curve $\lambda^2_{z}$ lies inside $C_{vz}$ and connects $p_{vz'}$ with $p_{vz}$. Curves $\lambda^0_{z}$ and $\lambda^2_{z}$ are constructed as in Lemma~\ref{le:outerplanar-draw}. Curve $\lambda_{u}(G)$ is constructed as follows.

\begin{itemize}
\item If $Z>2$, then $\lambda_{u}(G)$ consists of curves $\lambda^0_{u},\dots,\lambda^3_{u}$. Curve $\lambda^0_{u}$ lies inside $C_{uz}$ and connects $p_{uz}$ with $z_2$, which is internal to $P_z$ since $Z>2$; curve $\lambda^1_{u}$ coincides with path $(z_2,\dots,z_{Z-1})$ (the path consists of a single vertex
 if $Z=3$); curve $\lambda^2_{u}$ lies inside $C_{vz}$ and connects $z_{Z-1}$ with $p_{vz'}$; finally, curve $\lambda^3_{u}$ coincides with $\lambda_{v'}(H)$. Curves $\lambda^0_{u}$ and $\lambda^2_{u}$ are constructed as in Lemma~\ref{le:outerplanar-draw}.
\item If $Z=2$, then $\lambda_{u}(G)$ consists of curves $\lambda^0_{u},\lambda^1_{u},\lambda^2_{u}$. Curve $\lambda^0_{u}$ lies inside $C_{uz}$ and connects $p_{uz}$ with $p_{zz'}$; curve $\lambda^1_{u}$ lies inside $C_{vz}$ and connects $p_{zz'}$ with $p_{vz'}$; curve $\lambda^2_{u}$ coincides with $\lambda_{v'}(H)$. Curves $\lambda^0_{u}$ and $\lambda^1_{u}$ are constructed as in Lemma~\ref{le:outerplanar-draw}.
\end{itemize}

This completes the construction of $\lambda_{u}(G)$, $\lambda_{v}(G)$, and $\lambda_{z}(G)$. Since these curves lie in the interior of $G$ and since their end-points are incident to the outer face of $G$, they are proper. We now prove that they are good and pass through many vertices of $G$.

\begin{lemma} \label{le:plane-good}
Curves $\lambda_{u}(G)$, $\lambda_{v}(G)$, and $\lambda_{z}(G)$ are good.
\end{lemma} 

\begin{proof}
We prove that $\lambda_{u}(G)$ is good by induction on $m$; the proof for $\lambda_{v}(G)$ and $\lambda_{z}(G)$ is analogous. If $m\leq 1$ the statement is trivial. If $m>1$, then the central vertex $w$ of $G$ is of one of types B--D. 

If $w$ is of type C or D, then $\lambda_{u}(G)$ is composed of the three curves $\lambda_v(G_1)$, $\lambda_w(G_3)$, and $\lambda_z(G_2)$, each of which is good by induction. By construction, $\lambda_{u}(G)$ intersects edges $(u,v)$, $(v,w)$, $(z,w)$, and $(u,z)$ at points $p_{uv}$, $p_{vw}$, $p_{zw}$, and $p_{uz}$, respectively, and does not intersect edges $(v,z)$ and $(u,w)$ at all. Consider an edge $e$ internal to $G_1$. Curves $\lambda_w(G_3)$ and $\lambda_z(G_2)$ have no intersection with the region of the plane inside cycle $(u,v,w)$; further, $\lambda_{u}(G)$ does not pass through $u$, $v$, or $w$. Hence, $\lambda_{u}(G)$ contains $e$ or intersects at most once $e$, given that $\lambda_v(G_1)$ is good. Analogously, $\lambda_{u}(G)$ contains or intersects at most once every internal edge of $G_2$ and $G_3$. 

Assume now that $w$ is of type B. We prove that, for every edge $e$ of $G$, curve $\lambda_{u}(G)$ either contains $e$ or intersects $e$ at most once.


\begin{itemize}
\item By construction, $\lambda_{u}(G)$ intersects each of $(u,v)$, $(u,z)$, $(v,z)$, $(u',v')$, $(u',z')$, and  $(v',z')$ at most once. Also, $\lambda_{u}(G)$ has no intersection with any edge of the path $P_u$.
\item Consider an edge $e$ internal to $H$. The curves that compose $\lambda_{u}(G)$ and that lie inside $C_{uv}$, $C_{uz}$, or $C_{vz}$, or that coincide with a subpath of $P_v$ or $P_z$ have no intersection with the region of the plane inside cycle $(u',v',z')$; further, $\lambda_{u}(G)$ does not pass through $u'$, $v'$, or $z'$. Hence, $\lambda_{u}(G)$ contains $e$ or intersects at most once $e$, given that $\lambda_{u'}(H)$, $\lambda_{v'}(H)$, and $\lambda_{z'}(H)$ are good. 
\item Consider an edge $e=(v_j,v_{j+1})\in P_v$ (the argument for the edges in $P_z$ is analogous). If $\lambda_{u}(G)$ has no intersection with $P_v$, then it has no intersection with $e$. If $\lambda_{u}(G)$ intersects $P_v$ and $V>2$, then it contains $e$ (if $2\leq j\leq V-2$), or it intersects $e$ only at $v_{j+1}$ (if $j=1$), or it intersects $e$ only at $v_{j}$ (if $j=V-1$). Finally, if $\lambda_{u}(G)$ intersects $P_v$ and $V=2$, then $\lambda_{u}(G)$ properly crosses $e$ at $p_{vv'}$. 
\item We prove that $\lambda_{u}(G)$ intersects at most once the edges inside $C_{uv}$ (the argument for the edges inside $C_{uz}$ or $C_{vz}$ is analogous). Recall that, since $P_u$ and $P_v$ are induced, every edge inside $C_{uv}$ connects a vertex of $P_u$ and a vertex of $P_v$. Assume that $\lambda_{u}(G)$ contains a curve $\lambda^0_{u}$ inside $C_{uv}$ that connects $p_{uv}$ with $v_2$, a curve $\lambda^1_{u}$ that coincides with path $(v_2,\dots,v_{V-1})$, and a curve $\lambda^2_{u}$ inside $C_{uv}$ that connects $v_{V-1}$ with $p_{u'v'}$, as in Fig.~\ref{fig:plane3-tree-construction}(b); all the other cases are simpler to handle. 

\begin{itemize}
\item Consider any edge $e$ incident to $v_1$ inside $C_{uv}$. Curve $\lambda^0_{u}$ intersects $e$ once -- in fact the end-points of $\lambda^0_{u}$ alternate with those of $e$ along $C_{uv}$, hence $\lambda^0_{u}$ intersects $e$; moreover, $\lambda^0_{u}$ and $e$ do not intersect more than once by Lemma~\ref{le:outerplanar-draw}. Path $(v_2,\dots,v_{V-1})$, and hence curve $\lambda^1_{u}$ that coincides with it, has no intersection with $e$, since the end-vertices of $e$ are not in $v_2,\dots,v_{V-1}$. Further, curve $\lambda^2_{u}$ has no intersection with $e$ -- in fact the end-points of $\lambda^2_{u}$ do not alternate with those of $e$ along $C_{uv}$, hence if $\lambda^2_{u}$ and $e$ intersected, they would intersect at least twice, which is not possible by Lemma~\ref{le:outerplanar-draw}. Thus, $\lambda_{u}(G)$ intersects $e$ once.
\item Analogously, every edge $e$ incident to $v_V$ inside $C_{uv}$ has no intersection with $\lambda^0_{u}$, no intersection with $\lambda^1_{u}$, and one intersection with $\lambda^2_{u}$, hence $\lambda_{u}(G)$ intersects $e$ once. 
\item Finally, consider any edge $e$ incident to $v_j$, with $2\leq j\leq V-1$. Curve $\lambda^0_{u}$ and $\lambda^2_{u}$ have no intersection with $e$ -- in fact the end-points of each of these curves do not alternate with those of $e$ along $C_{uv}$, hence each of these curves does not intersect $e$ by Lemma~\ref{le:outerplanar-draw}. Further, $\lambda^1_{u}$ contains an end-vertex of $e$ and thus it intersects $e$ once. It follows that $\lambda_{u}(G)$ intersects $e$ once.
\end{itemize}
\end{itemize}

This concludes the proof of the lemma.
\end{proof}

We introduce three parameters. Let $s(G)$ be the total number of vertices of $G$ curves $\lambda_{u}(G)$, $\lambda_{v}(G)$, and $\lambda_{z}(G)$ pass through, counting each vertex with a multiplicity equal to the number of curves that pass through it. Further, let $x(G)$ be the number of internal vertices of type B none of $\lambda_{u}(G)$, $\lambda_{v}(G)$, and $\lambda_{z}(G)$ passes through. Finally, let $h(G)$ be the number of B-chains of $G$. We have the following inequalities.

\begin{lemma} \label{le:inequalities}
The following hold true if $m\geq 1$:
\begin{enumerate}
\item[(1)] $a(G)+b(G)+c(G)+d(G)=m$;
\item[(2)] $a(G)=c(G)+2 d(G)+1$;
\item[(3)] $h(G)\leq 2c(G)+3 d(G)+1$;
\item[(4)] $x(G)\leq b(G)$;
\item[(5)] $x(G)\leq 3h(G)$; and
\item[(6)] $s(G)\geq 3a(G)+b(G)-x(G)$.
\end{enumerate}
\end{lemma} 

\begin{proof}
{\bf (1)} $a(G)+b(G)+c(G)+d(G)=m$. This equality follows from the fact that every internal vertex of $G$ is of one of types A--D. 

\noindent{\bf (2)} $a(G)=c(G)+2 d(G)+1$. We use induction on $m$. If $m=1$ the statement is easily proved, as then the only internal vertex $w$ of $G$ is of type A, hence $a(G)=1$ and $c(G)=d(G)=0$. If $m>1$, then the central vertex $w$ of $G$ is of one of types B--D. 
%
%

	Suppose first that $w$ is of type B. Also, suppose that $G_1$ has internal vertices; the other cases are analogous. 
	Since $w$ is of type B, we have $a(G)=a(G_1)$, $c(G)=c(G_1)$, and $d(G)=d(G_1)$. Hence, $a(G)=a(G_1)=c(G_1)+2 d(G_1)+1=c(G)+2 d(G)+1$; the second equality holds by induction.
	
	Suppose next that $w$ is of type C. Also, suppose that $G_1$ and $G_2$ have internal vertices; the other cases are analogous. Since $w$ is of type C, we have $a(G)=a(G_1)+a(G_2)$, $c(G)=c(G_1)+c(G_2)+1$, and $d(G)=d(G_1)+d(G_2)$. Hence, $a(G)=a(G_1)+a(G_2)=(c(G_1)+2 d(G_1)+1)+(c(G_2)+2 d(G_2)+1)=(c(G_1)+c(G_2)+1)+2(d(G_1)+d(G_2))+1=c(G)+2 d(G)+1$; the second equality holds by induction.
	 
	Suppose finally that $w$ is of type D. Then we have $a(G)=a(G_1)+a(G_2)+a(G_3)$, $c(G)=c(G_1)+c(G_2)+c(G_3)$, and $d(G)=d(G_1)+d(G_2)+d(G_3)+1$. Hence, $a(G)=a(G_1)+a(G_2)+a(G_3)=(c(G_1)+2 d(G_1)+1)+(c(G_2)+2 d(G_2)+1)+(c(G_3)+2 d(G_3)+1)=(c(G_1)+c(G_2)+c(G_3))+2(d(G_1)+d(G_2)+d(G_3)+1)+1=c(G)+2 d(G)+1$; the second equality holds by induction. 
%

{\bf (3)} $h(G)\leq 2c(G)+3 d(G)+1$. We use induction on $m$. If $m=1$, then the only internal vertex $w$ of $G$ is of type A, hence $h(G)=0<1=2c(G)+3 d(G)+1$. If $m>1$, then the central vertex $w$ of $G$ is of one of types B--D. 

%
	Suppose first that $w$ is of type C. Also, suppose that $G_1$ and $G_2$ have internal vertices; the other cases are analogous. Since $w$ is of type C, we have $h(G)=h(G_1)+h(G_2)$, $c(G)=c(G_1)+c(G_2)+1$, and $d(G)=d(G_1)+d(G_2)$. Hence, $h(G)=h(G_1)+h(G_2)\leq(2c(G_1)+3d(G_1)+1)+(2c(G_2)+3d(G_2)+1)=2(c(G_1)+c(G_2)+1)+3(d(G_1)+d(G_2))=
2c(G)+3d(G)<2c(G)+3d(G)+1$; the second inequality holds by induction.

Second, if $w$ is of type D, we have $h(G)=h(G_1)+h(G_2)+h(G_3)$, $c(G)=c(G_1)+c(G_2)+c(G_3)$, and $d(G)=d(G_1)+d(G_2)+d(G_3)+1$. Hence, $h(G) = h(G_1)+h(G_2)+h(G_3) \leq (2c(G_1)+3d(G_1)+1)+(2c(G_2)+3d(G_2)+1)+(2c(G_3)+3d(G_3)+1)=2(c(G_1)+c(G_2)+c(G_3))+3(d(G_1)+d(G_2)+d(G_3)+1)=
2c(G)+3d(G)<2c(G)+3d(G)+1$; the second inequality holds by induction.

%
Finally, suppose that $w$ is of type B. Then $w_1=w$ is the first vertex of a B-chain $w_1,\dots,w_i$ of $G$. Recall that $H$ is the only plane $3$-tree child of $w_i$ that has internal vertices. Let $x$ be the central vertex of $H$. By the maximality of $w_1,\dots,w_i$, we have that $x$ is not of type B, hence $x$ is of type A, C, or D. If $x$ is of type A, we have $h(G)=1$, $c(G)=d(G)=0$, hence $h(G)=1=2c(G)+3d(G)+1$.

If $x$ is of type C, then let $L_1$ and $L_2$ be the children of $H$ containing internal vertices. We have $h(G)=h(L_1)+h(L_2)+1$, $c(G)=c(L_1)+c(L_2)+1$, and $d(G)=d(L_1)+d(L_2)$. Thus, $h(G)=h(L_1)+h(L_2)+1\leq (2c(L_1)+3d(L_1)+1) + (2c(L_2)+3d(L_2)+1)+1=2(c(L_1)+c(L_2)+1)+3(d(L_1)+d(L_2))+1=2c(G)+3d(G)+1$; the second inequality holds by induction.

Finally, if $x$ is of type D, then let $L_1$, $L_2$, and $L_3$ be the children of $H$. We have $h(G)=h(L_1)+h(L_2)+h(L_3)+1$, $c(G)=c(L_1)+c(L_2)+c(L_3)$, and $d(G)=d(L_1)+d(L_2)+d(L_3)+1$. Thus, $h(G)=h(L_1)+h(L_2)+h(L_3)+1\leq (2c(L_1)+3d(L_1)+1) + (2c(L_2)+3d(L_2)+1)+(2c(L_3)+3d(L_3)+1)+1=2(c(L_1)+c(L_2)+c(L_3))+3(d(L_1)+d(L_2)+d(L_3)+1)+1=2c(G)+3d(G)+1$; the second inequality holds by induction.

{\bf (4)} $x(G)\leq b(G)$. This inequality follows from the fact that $x(G)$ is the number of vertices of type B of $G$ none of $\lambda_{u}(G)$, $\lambda_v(G)$, and $\lambda_z(G)$ passes through, hence this number cannot be larger than the number of vertices of type B of $G$. 

{\bf (5)} $x(G)\leq 3h(G)$. Every internal vertex of $G$ of type B belongs to a B-chain of $G$. Further, for every B-chain $w_1,w_2,\dots,w_i$ of $G$, curves $\lambda_{u}(G)$, $\lambda_{v}(G)$, and $\lambda_{z}(G)$ pass through all of $w_1,w_2,\dots,w_i$, except for at most three vertices $u'=u_U$, $v'=v_V$, and $z'=z_Z$ (note that, in the description of the construction of $\lambda_{u}(G)$, $\lambda_{v}(G)$, and $\lambda_{z}(G)$ if $w$ is of type B, vertices $u$, $v$, and $z$ are not among $w_1,w_2,\dots,w_i$). Thus, the number $x(G)$ of vertices of type B none of $\lambda_{u}(G)$, $\lambda_{v}(G)$, and $\lambda_{z}(G)$ passes through is at most three times the number $h(G)$ of B-chains of $G$.    

{\bf (6)} $s(G)\geq 3a(G)+b(G)-x(G)$. We use induction on $m$. If $m=1$ then the only internal vertex $w$ of $G$ is of type A, hence $a(G)=1$ and $b(G)=x(G)=0$. Further, by construction, each of $\lambda_{u}(G)$, $\lambda_{v}(G)$, and $\lambda_{z}(G)$ passes through $w$, hence $s(G)=3$. Thus, $s(G)=3=3a(G)+b(G)-x(G)$. If $m>1$, then the central vertex $w$ of $G$ is of one of types B--D. 

Suppose first that $w$ is of type C. Also, suppose that $G_1$ and $G_2$ have internal vertices; the other cases are analogous. Since $w$ is of type C, we have $a(G)=a(G_1)+a(G_2)$, $b(G)=b(G_1)+b(G_2)$, and $x(G)=x(G_1)+x(G_2)$. By construction, curves $\lambda_{u}(G)$, $\lambda_{v}(G)$, and $\lambda_{z}(G)$ contain all of $\lambda_{u}(G_1)$, $\lambda_{v}(G_1)$, $\lambda_{w}(G_1)$, $\lambda_{u}(G_2)$, $\lambda_{z}(G_2)$, and $\lambda_{w}(G_2)$. It follows that $s(G)=s(G_1)+s(G_2)\geq (3a(G_1)+b(G_1)-x(G_1))+(3a(G_2)+b(G_2)-x(G_2)) = 3(a(G_1)+a(G_2))+(b(G_1)+b(G_2))-(x(G_1)+x(G_2))=3a(G)+b(G)-x(G)$; the second inequality follows by induction.

Suppose next that $w$ is of type D. Then we have $a(G)=a(G_1)+a(G_2)+a(G_3)$, $b(G)=b(G_1)+b(G_2)+b(G_3)$, and $x(G)=x(G_1)+x(G_2)+x(G_3)$. By construction, curves $\lambda_{u}(G)$, $\lambda_{v}(G)$, and $\lambda_{z}(G)$ contain all of $\lambda_{u}(G_1)$, $\lambda_{v}(G_1)$, $\lambda_{w}(G_1)$, $\lambda_{u}(G_2)$, $\lambda_{z}(G_2)$, $\lambda_{w}(G_2)$, $\lambda_{v}(G_3)$, $\lambda_{z}(G_3)$, and $\lambda_{w}(G_3)$. It follows that $s(G)=s(G_1)+s(G_2)+s(G_3)\geq (3a(G_1)+b(G_1)-x(G_1))+(3a(G_2)+b(G_2)-x(G_2)) +(3a(G_3)+b(G_3)-x(G_3)) = 3(a(G_1)+a(G_2)+a(G_3))+(b(G_1)+b(G_2)+b(G_3))-(x(G_1)+x(G_2)+x(G_3))=3a(G)+b(G)-x(G)$; the second inequality follows by induction.

Suppose finally that $w$ is of type B. Then $w_1=w$ is the first vertex of a B-chain $w_1,\dots,w_i$ of $G$ and $H$ is the only plane $3$-tree child of $w_i$ that has internal vertices. Every internal vertex of $G$ of type A is internal to $H$, hence $a(G)=a(H)$. Every internal vertex of $G$ of type B is either an internal vertex of $H$ of type B, or is one among $w_1,\dots,w_i$; hence $b(G)=b(H)+i$. Since $\lambda_{u}(G)$, $\lambda_{v}(G)$, and $\lambda_{z}(G)$ contain all of $\lambda_{u'}(H)$, $\lambda_{v'}(H)$, and $\lambda_{z'}(H)$, we have that $s(G)$ is greater than or equal to $s(H)$ plus the number of vertices among $w_1,\dots,w_i$ curves $\lambda_{u}(G)$, $\lambda_{v}(G)$, and $\lambda_{z}(G)$ pass through; for the same reason, $x(G)$ is equal to $x(H)$ plus the number of vertices among $w_1,\dots,w_i$ none of $\lambda_{u}(G)$, $\lambda_{v}(G)$, and $\lambda_{z}(G)$ passes through. By construction, $\lambda_{u}(G)$, $\lambda_{v}(G)$, and $\lambda_{z}(G)$ do not pass through at most three vertices among $w_1,\dots,w_i$, hence $x(G)\leq x(H)+3$ and $s(G)\geq s(H)+i-3$. Thus, we have $s(G)\geq s(H)+i-3 \geq 3a(H)+b(H)-x(H) + i-3 = 3a(H)+(b(H)+i)-(x(H)+3)\geq 3a(G)+b(G)-x(G)$; the second inequality follows by induction.
\end{proof}

Lemma~\ref{le:inequalities} can be used to prove that one of $\lambda_{u}(G)$, $\lambda_{v}(G)$, and $\lambda_{z}(G)$ passes through many vertices of $G$. Let $k$ be a parameter to be determined later. 

If $a(G)\geq km$, then by (4) and (6) we get $s(G)\geq 3a(G)\geq 3km$.

If $a(G)< km$, by (1) and (6) we get $s(G)\geq 3a(G) + (m -a(G)- c(G) - d(G)) - x(G)$, which by (5) becomes $s(G)\geq m + 2a(G) - c(G) - d(G) - 3h(G)$. Using (2) and (3) we get $s(G)\geq m + 2 (c(G)+2d(G)+1) -c(G)-d(G)-3(2c(G)+3d(G)+1)= m - 5c(G)-6d(G)-1$. Again by (2) and by hypothesis we get $c(G)+2d(G)+1<km$, thus $5c(G)+6d(G)+1<5c(G)+10d(G)+5<5km$. Hence, $s(G)\geq m-5km$.

Let $k=\frac{1}{8}$. We get $3km=m-5km=\frac{3m}{8}$, thus $s(G)\geq \frac{3m}{8}$ both if $a(G)\geq \frac{m}{8}$ and if $a(G)< \frac{m}{8}$. It follows that one of $\lambda_u(G)$, $\lambda_v(G)$, and $\lambda_z(G)$ is a proper good curve passing through $\lceil\frac{n-3}{8}\rceil$ internal vertices of $G$. This concludes the proof of Theorem~\ref{th:3-trees}.

We now prove that not only a \cfsl\ of a plane $3$-tree exists with $\lceil\frac{n-3}{8}\rceil$ collinear vertices, but the geometric placement of these vertices can be arbitrarily prescribed, as long as it satisfies an ordering constraint. Since every plane graph of treewidth at most three is a subgraph of a plane $3$-tree~\cite{kv-npp3t-12}, this implies that every plane graph of treewidth at most three has a free collinear set with $\lceil\frac{n-3}{8}\rceil$ vertices.

\begin{theorem} \label{th:csvfcs} 
Every collinear set in a plane $3$-tree is also a free collinear set. 
\end{theorem} 

Let $G$ be an $n$-vertex plane $3$-tree with external vertices $u$, $v$, and $z$ in this counter-clockwise order along cycle $(u,v,z)$. Consider any \cfsl\ $\Psi$ of $G$ and a horizontal line $\ell$. Label each vertex of $G$ as $\uparrow$, $\downarrow$, or $=$ according to whether it lies above, below, or on $\ell$, respectively; let $S$ be the set of vertices labeled $=$. Let $E_{\ell}$ be the set of edges of $G$ that properly cross $\ell$ in $\Psi$; thus, the edges in $E_{\ell}$ have one end-vertex labeled $\uparrow$ and one end-vertex labeled $\downarrow$. Let $<_{\Psi}$ be the total ordering of $S \cup E_{\ell}$ corresponding to the left-to-right order in which the vertices in $S$ and the crossing points between the edges in $E_{\ell}$ and $\ell$ appear along $\ell$ in $\Psi$. 

Let $X$ be any set of $|S| + |E_{\ell}|$ distinct points on $\ell$. Each element in $S \cup E_{\ell}$ is associated with a point in $X$: The $i$-th element of $S \cup E_{\ell}$, where the elements in $S \cup E_{\ell}$ are ordered according to $<_{\Psi}$, is associated with the $i$-th point of $X$, where the points in $X$ are in left-to-right order along $\ell$. Denote by $X_S$ and $X_E$ the subsets of the points in $X$ associated to the vertices in $S$ and to the edges in $E_{\ell}$, respectively; also, denote by $q_x$ the point in $X$ associated with a vertex $x\in S$ and by $q_{xy}$ the point in $X$ associated with an edge $(x,y)\in E_{\ell}$. 

We have the following lemma, which implies Theorem~\ref{th:csvfcs}.

\begin{lemma} \label{le:conformal}
There exists a \cfsl\ $\Gamma$ of $G$ such that: 
\begin{enumerate}
\item[(1)] $\Gamma$ {\em respects the labeling} -- every vertex labeled $\uparrow$, $\downarrow$, or $=$ is above, below, or on $\ell$, respectively; and 
\item[(2)] $\Gamma$ {\em respects the ordering} -- every vertex in $S$ is placed at its associated point in $X_S$ and every edge in $E_{\ell}$ crosses $\ell$ at its associated point in $X_E$.
\end{enumerate}
\end{lemma} 

\begin{proof}
The proof is by induction on $n$ and relies on a stronger inductive hypothesis, namely that $\Gamma$ can be constructed for any \cfsl\ $\Delta$ of cycle $(u,v,z)$ such that: 

\begin{itemize}
\item[(i)] the vertices $p_u$, $p_v$, and $p_z$ of $\Delta$ representing $u$, $v$, and $z$ appear in this counter-clockwise order along $\Delta$; 
\item[(ii)] $\Delta$ respects the labeling -- each of $u$, $v$, and $z$ is above, below, or on $\ell$ if it has label $\uparrow$, $\downarrow$, or $=$, respectively; and 
\item[(iii)] $\Delta$ respects the ordering -- every vertex in $\{u,v,z\}\cap S$ lies at its associated point in $X_S$ and every edge in $\{(u,v),(u,z),(v,z)\}\cap E_{\ell}$ crosses $\ell$ at its associated point in $X_E$.
\end{itemize}



In the base case $n=3$. Let $\Delta$ be any \cfsl\ of cycle $(u,v,z)$ satisfying properties (i)--(iii). Define $\Gamma=\Delta$; then $\Gamma$ is a \cfsl\ of $G$ that respects the labeling and the ordering since $\Delta$ satisfies properties (i)--(iii). 

Now assume that $n>3$; let $w$ be the central vertex of $G$, and let $G_1$, $G_2$, and $G_3$ be its children, whose outer faces are delimited by cycles $(u,v,w)$, $(u,z,w)$, and $(v,z,w)$, respectively. We distinguish some cases according to the labeling of $u$, $v$, $z$, and $w$. In every case we draw $w$ at a point $p_w$ and we draw straight-line segments from $p_w$ to $p_u$, $p_v$, and $p_z$, obtaining triangles $\Delta_1=(p_u,p_v,p_w)$, $\Delta_2=(p_u,p_z,p_w)$, and $\Delta_3=(p_v,p_z,p_w)$. We then use induction to construct \cfsl s of $G_1$, $G_2$, and $G_3$ in which the cycles $(u,v,w)$, $(u,z,w)$, and $(v,z,w)$ delimiting their outer faces are represented by $\Delta_1$, $\Delta_2$, and $\Delta_3$, respectively. Thus, we only need to ensure that each of $\Delta_1$, $\Delta_2$, and $\Delta_3$ satisfies properties (i)--(iii). In particular, property (i) is satisfied as long as $p_w$ is in the interior of $\Delta$; property (ii) is satisfied as long as $p_w$ respects the labeling; and property (iii) is satisfied as long as $p_w=q_w$, if $w\in S$, and each edge in $\{(u,w),(v,w),(z,w)\}\cap E_{\ell}$ crosses $\ell$ at its associated point, if $w\notin S$.



If all of $u$, $v$, and $z$ have labels in the set $\{\uparrow,=\}$, then all the internal vertices of $G$ have label $\uparrow$, by the planarity of $\Psi$, and the interior of $\Delta$ is above $\ell$. Let $p_w$ be any point in the interior of $\Delta$ (ensuring properties (i)--(ii) for $\Delta_1$, $\Delta_2$, and $\Delta_3$). Also, $w\notin S$ and $(u,w),(v,w),(z,w)\notin E_{\ell}$, thus property (iii) is satisfied for $\Delta_1$, $\Delta_2$, and $\Delta_3$.


The case in which all of $u$, $v$, and $z$ have labels in the set $\{\downarrow,=\}$ is symmetric. If none of these cases applies, we can assume w.l.o.g. that $u$ has label $\uparrow$ and $v$ has label $\downarrow$. 

\begin{itemize}
\item Suppose that $z$ has label $=$. Since $u$ has label $\uparrow$, $v$ has label $\downarrow$, and $(u,v,z)$ has this counter-clockwise orientation in $G$, edge $(u,v)$ and vertex $z$ are respectively the first and the last element in $S \cup E_{\ell}$ according to $<_{\Psi}$. Since $\Delta$ satisfies properties (i)-(iii), points $q_{uv}$ and $q_z$ are respectively the leftmost and the rightmost point in $X$; hence all the points in $X-\{q_{uv},q_z\}$ are in the interior of $\Delta$.

\begin{figure}[htb]
\begin{center}
\begin{tabular}{c c c c c}
\mbox{\includegraphics[scale=.4]{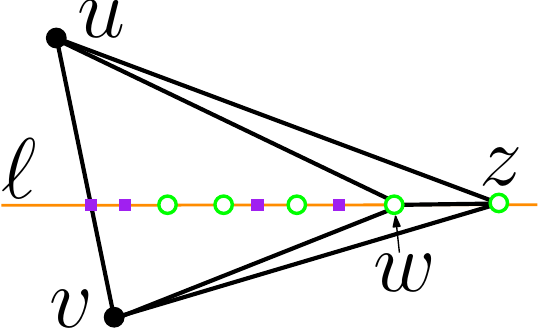}} \hspace{1mm} &
\mbox{\includegraphics[scale=.4]{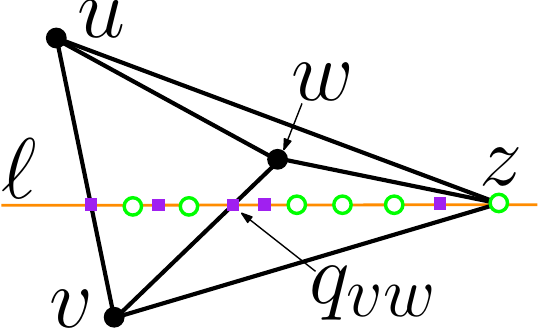}} \hspace{1mm} &
\mbox{\includegraphics[scale=.4]{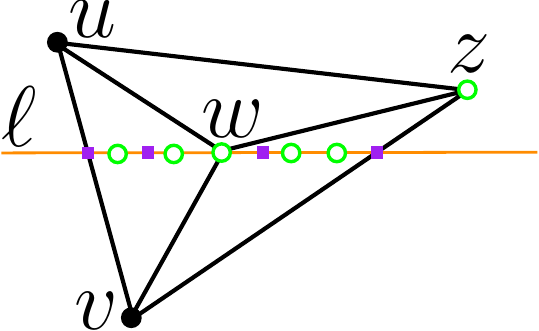}} \hspace{1mm} &
\mbox{\includegraphics[scale=.4]{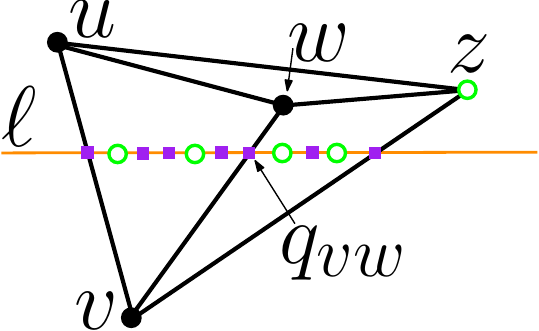}} \hspace{1mm} &
\mbox{\includegraphics[scale=.4]{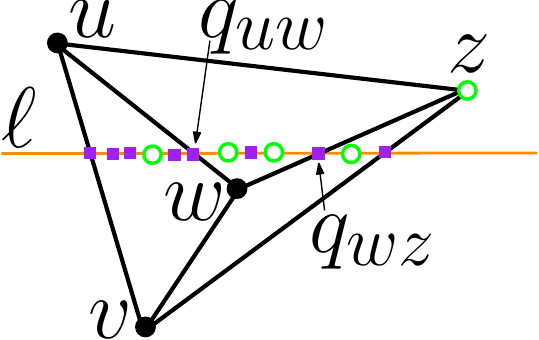}} \\
(a) \hspace{1mm} & (b) \hspace{1mm} & (c) \hspace{1mm} & (d) \hspace{1mm} & (e)
\end{tabular}
\caption{Cases for the proof of Lemma~\ref{le:conformal}. Line $\ell$ is orange, points in $X_S$ are green, and points in $X_E$ are purple. (a) $z$ and $w$ have label $=$; (b) $z$ has label $=$ and $w$ has label $\uparrow$; (c) $z$ has label $\uparrow$ and $w$ has label $=$; (d) $z$ and $w$ have label $\uparrow$; and (e) $z$ has label $\uparrow$ and $w$ has label $\downarrow$.}
\label{fig:universal}
\end{center}
\end{figure}

\begin{itemize}
\item If $w$ has label $=$, as in Fig.~\ref{fig:universal}(a), then $w$ is the last but one element in $S \cup E_{\ell}$ according to $<_{\Psi}$, by the planarity of $\Psi$ (note that edge $(w,z)$ lies on $\ell$). Since $\Delta$ satisfies (i)--(iii), point $q_w$ is the rightmost point in $X-\{q_z\}$. Let $p_w=q_w$ (ensuring properties (i)--(ii) for $\Delta_1$, $\Delta_2$, and $\Delta_3$). Then $w$ is at $q_w$ and $(u,w),(v,w),(z,w)\notin E_{\ell}$ (ensuring property (iii) for $\Delta_1$, $\Delta_2$, and $\Delta_3$).  

\item If $w$ has label $\uparrow$, as in Fig.~\ref{fig:universal}(b), then edge $(v,w)$ comes after edge $(u,v)$ and before vertex $z$ in $S \cup E_{\ell}$ according to $<_{\Psi}$, since $(v,w)$ is an internal edge of $G$ and $\Psi$ is planar. Since $\Delta$ satisfies (i)--(iii), point $q_{vw}$ is between $q_{uv}$ and $q_z$ on $\ell$. Draw a half-line $h$ starting at $v$ through $q_{vw}$ and let $p_w$ be any point in the interior of $\Delta$ (ensuring property (i) for $\Delta_1$, $\Delta_2$, and $\Delta_3$) after $q_{vw}$ on $h$ (ensuring property (ii) for $\Delta_1$, $\Delta_2$, and $\Delta_3$). Then $w\notin S$, $(u,w),(z,w)\notin E_{\ell}$, and the crossing point between $(v,w)$ and $\ell$ is $q_{vw}$ (ensuring property (iii) for $\Delta_1$, $\Delta_2$, and $\Delta_3$).  

\item The case in which $w$ has label $\downarrow$ is symmetric to the previous one.
\end{itemize}

\item Assume now that $z$ has label $\uparrow$. Since $u$ and $z$ have label $\uparrow$, since $v$ has label $\downarrow$, and since $(u,v,z)$ has this counter-clockwise orientation in $G$, edges $(u,v)$ and $(v,z)$ are respectively the first and the last element in $S \cup E_{\ell}$ according to $<_{\Psi}$. Since $\Delta$ satisfies properties (i)-(iii), points $q_{uv}$ and $q_{vz}$ are respectively the leftmost and the rightmost point in $X$; thus all the points in $X-\{q_{uv},q_{vz}\}$ are in the interior of $\Delta$.

\begin{itemize}
\item If $w$ has label $=$, as in Fig.~\ref{fig:universal}(c), then vertex $w$ comes after edge $(u,v)$ and before edge $(v,z)$ in $S \cup E_{\ell}$ according to $<_{\Psi}$, since $w$ is an internal vertex of $G$ and $\Psi$ is planar. Since $\Delta$ satisfies (i)--(iii), $q_w$ is between $q_{uv}$ and $q_{vz}$ on $\ell$. Let $p_w=q_w$ (ensuring properties (i)--(ii) for $\Delta_1$, $\Delta_2$, and $\Delta_3$). Then $w$ is at $q_w$ and $(u,w),(v,w),(z,w)\notin E_{\ell}$ (ensuring property (iii) for $\Delta_1$, $\Delta_2$, and $\Delta_3$).  

\item If $w$ has label $\uparrow$, as in Fig.~\ref{fig:universal}(d), then edge $(v,w)$ comes after edge $(u,v)$ and before edge $(v,z)$ in $S \cup E_{\ell}$ according to $<_{\Psi}$, since $(v,w)$ is an internal edge of $G$ and $\Psi$ is planar. Since $\Delta$ satisfies (i)--(iii), point $q_{vw}$ is between $q_{uv}$ and $q_{vz}$ on $\ell$. Draw a half-line $h$ starting at $v$ through $q_{vw}$ and let $p_w$ be any point in the interior of $\Delta$ (ensuring property (i) for $\Delta_1$, $\Delta_2$, and $\Delta_3$) after $q_{vw}$ on $h$ (ensuring property (ii) for $\Delta_1$, $\Delta_2$, and $\Delta_3$). Then $w\notin S$, $(u,w),(z,w)\notin E_{\ell}$, and the crossing point between $(v,w)$ and $\ell$ is $q_{vw}$ (ensuring property (iii) for $\Delta_1$, $\Delta_2$, and $\Delta_3$).

\item If $w$ has label $\downarrow$, as in Fig.~\ref{fig:universal}(e), then edges $(u,v)$, $(u,w)$, $(w,z)$, and $(v,z)$ come in this order in $S \cup E_{\ell}$ according to $<_{\Psi}$, since $(u,w)$ and $(w,z)$ are internal edges of $G$ and $\Psi$ is planar. Since $\Delta$ satisfies (i)--(iii), $q_{uv},q_{uw},q_{wz},q_{vz}$ appear in this left-to-right order on $\ell$. Let $p_w$ be the intersection point between the line through $u$ and $q_{uw}$ and the line through $z$ and $q_{wz}$ (ensuring property (ii) for $\Delta_1$, $\Delta_2$, and $\Delta_3$); note that $p_w$ is in the interior of $\Delta$ (ensuring property (i) for $\Delta_1$, $\Delta_2$, and $\Delta_3$). Then $w\notin S$, $(v,w)\notin E_{\ell}$, the crossing point between $(v,w)$ and $\ell$ is $q_{vw}$, and the crossing point between $(w,z)$ and $\ell$ is $q_{wz}$ (ensuring property (iii) for $\Delta_1$, $\Delta_2$, and $\Delta_3$).
\end{itemize}

\item The case in which $z$ has label $\downarrow$ is symmetric to the previous one.
\end{itemize}

This concludes the proof of the lemma.
\end{proof}

\section{Triconnected Cubic Planar Graphs} \label{le:cubic}

In this section we prove the following theorem.

\begin{theorem} \label{th:cubic}
Every $n$-vertex triconnected cubic plane graph admits a \cfsl\ with at least $\lceil \frac{n}{4}\rceil$ collinear vertices.
\end{theorem}

By Theorem~\ref{th:topology} it suffices to prove that, for every $n$-vertex triconnected cubic plane graph $G$, there exists a proper good curve $\lambda$ passing through at least $\lceil \frac{n}{4}\rceil$ vertices of $G$. 

Our proof will be by induction on $n$; Lemma~\ref{le:cubic-2connected} below states the inductive hypothesis satisfied by $\lambda$. In order to split the graph into subgraphs on which induction can be applied, we will use a structural decomposition that is derived from a paper by Chen and Yu~\cite{cy-lc3g-02}, who proved that every $n$-vertex triconnected planar graph contains a cycle passing through $\Omega(n^{\log_3 2})$ vertices. This decomposition applies to a graph class, called {\em strong circuit graph} in~\cite{cy-lc3g-02}, wider than triconnected cubic plane graphs. We will introduce the concept of {\em well-formed quadruple} in order to point out some properties of the graphs in this class. In particular, the inductive hypothesis will need to handle carefully the set (denoted by $X$ below) of degree-$2$ vertices of the graph, which have neighbors that are not in the graph at the current level of the induction; since $\lambda$ might pass through these neighbors, it has to avoid the vertices in $X$, in order to be good. Special conditions will be ensured for two vertices, denoted by $u$ and $v$ below, which work as link to the rest of the graph. 

We introduce some definitions. Consider a biconnected plane graph $G$. Given two external vertices $u$ and $v$ of $G$, we denote by $\tau_{uv}(G)$ (by $\beta_{uv}(G)$) the path composed of the vertices and edges encountered when walking along the boundary of the outer face of $G$ in clockwise (resp.\ counter-clockwise) direction from $u$ to $v$. An {\em intersection point} between an open curve $\lambda$ and $\beta_{uv}(G)$ (or $\tau_{uv}(G)$) is a point $p$ belonging to both $\lambda$ and $\beta_{uv}(G)$ (resp.\ $\tau_{uv}(G)$) such that, for any $\epsilon>0$, the part of $\lambda$ in the disk centered at $p$ with radius $\epsilon$ contains points not in $\beta_{uv}(G)$ (resp.\ $\tau_{uv}(G)$). If the end-vertices of $\lambda$ are in $\beta_{uv}(G)$ (or $\tau_{uv}(G)$), then we regard them as intersection points. An intersection point $p$ between $\lambda$ and $\beta_{uv}(G)$ (or $\tau_{uv}(G)$) is {\em proper} if, for any $\epsilon>0$, the part of $\lambda$ in the disk centered at $p$ with radius $\epsilon$ contains points in the outer face of $G$. 




Our proof of the existence of a proper good curve passing through $\lceil \frac{n}{4}\rceil$ vertices of $G$ is by induction on $n$. In order to make the induction work, we deal with the following setting. A quadruple $(G,u,v,X)$ is {\em well-formed} if it satisfies the following properties. 

\begin{itemize}
\item[(a)] $G$ is a biconnected subcubic plane graph;
\item[(b)] $u$ and $v$ are two distinct external vertices of $G$;
\item[(c)] $\delta_G(u)=\delta_G(v)=2$; 
\item[(d)] if edge $(u,v)$ exists, then it coincides with $\tau_{uv}(G)$; 
\item[(e)] for every separation pair $\{a,b\}$ of $G$ we have that $a$ and $b$ are external vertices of $G$ and at least one of them is an internal vertex of $\beta_{uv}(G)$; further, every non-trivial $\{a,b\}$-component of $G$ contains an external vertex of $G$ different from $a$ and $b$; and 
\item[(f)] $X=(x_1,\dots,x_m)$ is a (possibly empty) sequence of degree-$2$ vertices of $G$ in $\beta_{uv}(G)$, different from $u$ and $v$, and in this order along $\beta_{uv}(G)$ from $u$ to $v$.
\end{itemize}


We have the following main lemma (refer to Fig.~\ref{fig:main-lemma-cubic}).

\begin{figure}[htb]
\begin{center}
\mbox{\includegraphics[width=.4\textwidth]{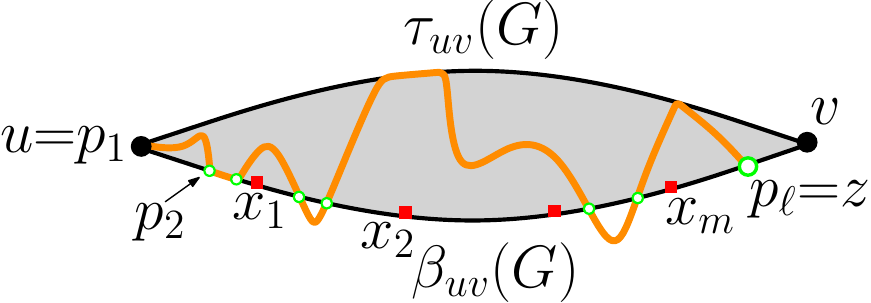}} \caption{Illustration for the statement of Lemma~\ref{le:cubic-2connected}. The gray region represents the interior of $G$. Curve $\lambda$ is orange, vertices in $X$ are red squares, intersection points between $\lambda$ and $\beta_{uv}(G)$ are green circles, and vertices $u$ and $v$ are black disks.}
\label{fig:main-lemma-cubic}
\end{center}
\end{figure}

\begin{lemma} \label{le:cubic-2connected}
Let $(G,u,v,X)$ be a well-formed quadruple. There exists a proper good curve $\lambda$ such that:
\begin{enumerate}
\item[(1)] $\lambda$ starts at $u$, does not pass through $v$, and ends at a point $z$ of $\beta_{uv}(G)$;
\item[(2)] $z$ is in the part of $\beta_{uv}(G)$ between $x_m$ and $v$ (if $X=\emptyset$, this condition is vacuous);
\item[(3)] let $u=p_1,p_2,\dots,p_\ell=z$ be the intersection points between $\lambda$ and $\beta_{uv}(G)$, ordered as they occur along $\lambda$; we have that $u=p_1,p_2,\dots,p_\ell=z,v$ appear in this order along $\beta_{uv}(G)$ (note that $z$ is the ``last'' intersection between $\lambda$ and $\beta_{uv}(G)$);
\item[(4)] $\lambda$ passes through no vertex in $X$ and all the vertices in $X$ are incident to $R_{G,\lambda}$; in particular, if $p_i$, $x_j$ and $p_{i+1}$ come in this order along $\beta_{uv}(G)$, then the part of $\lambda$ between $p_i$ and $p_{i+1}$ lies in the interior of $G$;
\item[(5)] $\lambda$ and $\tau_{uv}(G)$ have no proper intersection point; and
\item[(6)] let $L_{\lambda}$ and $N_{\lambda}$ be the subsets of vertices in $V(G)-X$ curve $\lambda$ passes through and does not pass through, respectively; each vertex in $N_{\lambda}$ can be charged to a vertex in $L_{\lambda}$ so that each vertex in $L_{\lambda}$ is charged with at most $3$ vertices and $u$ is charged with at most $1$ vertex.
\end{enumerate}
\end{lemma}


Before proceeding with the proof of Lemma~\ref{le:cubic-2connected}, we state and prove an auxiliary lemma that will be used repeatedly in the remainder of the section. 

\begin{lemma} \label{le:separation-bottom}
Let $(G,u,v,X)$ be a well-formed quadruple and let $\{a,b\}$ be a separation pair of $G$ with $a,b\in \beta_{uv}(G)$. The $\{a,b\}$-component $G_{ab}$ of $G$ containing $\beta_{ab}(G)$ either coincides with $\beta_{ab}(G)$ or consists of (see Fig.~\ref{fig:structure}):
\begin{itemize}
\item a path $P_0=(a,\dots,u_1)$ (possibly a single vertex $a=u_1$);
\item for $i=1,\dots,k$ with $k\geq 1$, a biconnected component $G_i$ of $G_{ab}$ that contains vertices $u_i$ and $v_i$ and such that $(G_i,u_i,v_i,X_i)$ is a well-formed quadruple, with $X_i=X\cap V(G_i)$; 
\item for $i=1,\dots,k-1$, a path $P_i=(v_i,\dots,u_{i+1})$, where $u_{i+1}\neq v_i$; and
\item a path $P_{k}=(v_{k},\dots,b)$ (possibly a single vertex $b=v_{k}$).
\end{itemize} 
\end{lemma}

\begin{figure}[htb]
\begin{center}
\mbox{\includegraphics[width=.8\textwidth]{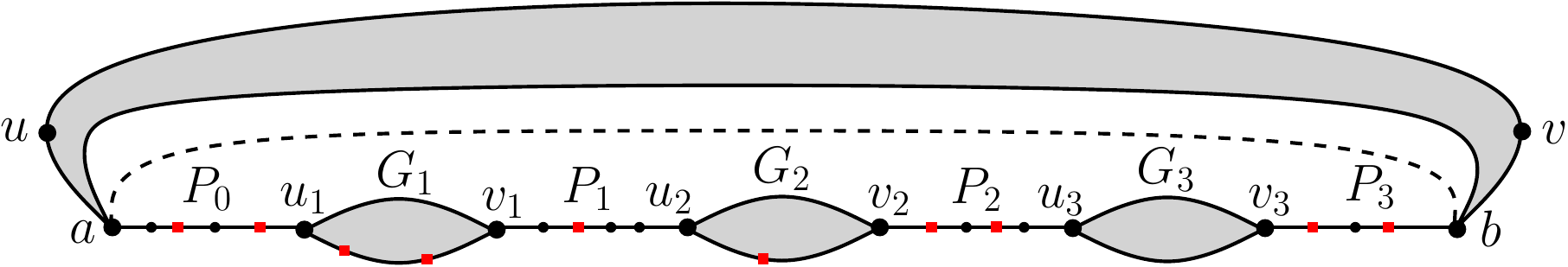}} \caption{Illustration for Lemma~\ref{le:separation-bottom} with $k=3$.}
\label{fig:structure}
\end{center}
\end{figure}

\begin{proof}
If $G$ contained more than two non-trivial $\{a,b\}$-components, then one of them would not contain any external vertex of $G$ different from $a$ and $b$, a contradiction to Property (e) of $(G,u,v,X)$. Thus, $G$ contains two non-trivial $\{a,b\}$-components, one of which is $G_{ab}$. Possibly, $G$ contains a trivial $\{a,b\}$-component which is an internal edge $(a,b)$ of $G$. The statement is proved by induction on the size of $G_{ab}$. 

In the base case, $G_{ab}$ is a path between $a$ and $b$ or is a biconnected graph. In the former case, $G_{ab}$ coincides with $\beta_{ab}(G)$ and the statement of the lemma follows. In the latter case, the statement of the lemma follows with $k=1$, $G_1=G_{ab}$, $P_0=a$, and $P_k=b$, as long as $(G_{ab},a,b,X_{ab})$ is a well-formed quadruple, where $X_{ab}=X\cap V(G_{ab})$. We now prove that this is indeed the case.

\begin{itemize}
\item Property (a): $G_{ab}$ is biconnected by hypothesis and subcubic since $G$ is subcubic. 
\item Property (b): $a$ and $b$ are external vertices of $G_{ab}$ as they are external vertices of $G$. 
\item Property (c): the degree of $a$ and $b$ in $G_{ab}$ is at least $2$, by the biconnectivity of $G_{ab}$, and at most $2$, since $G$ is subcubic and since $a$ and $b$ have a neighbor in the non-trivial $\{a,b\}$-component of $G$ different from $G_{ab}$. 
\item Property (d): if edge $(a,b)$ exists in $G$, then it forms a trivial $\{a,b\}$-component and it does not belong to $G_{ab}$, hence the property is trivially satisfied.
\item Property (e): consider any separation pair $\{a',b'\}$ of $G_{ab}$. If $G_{ab}$ contained more than two non-trivial $\{a',b'\}$-components, as in Fig.~\ref{fig:structure-proof}(a), then one of them would be a non-trivial $\{a',b'\}$-component of $G$ that contains no external vertex of $G$ different from $a'$ and $b'$, a contradiction to Property (e) of $(G,u,v,X)$. It follows that $G_{ab}$ contains two non-trivial $\{a',b'\}$-components $G'_{ab}$ and $G''_{ab}$. 


\begin{figure}[htb]
\begin{center}
\begin{tabular}{c c c}
\mbox{\includegraphics[width=.12\textwidth]{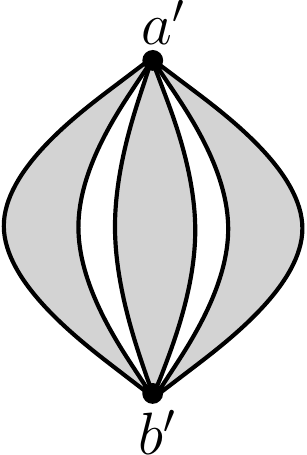}}\hspace{6mm} &
\mbox{\includegraphics[width=.2\textwidth]{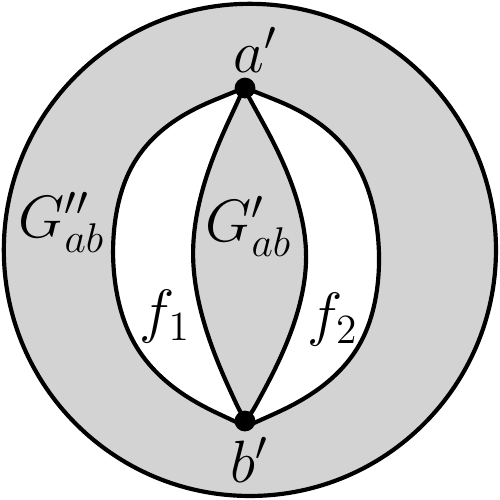}}\hspace{6mm} &
\mbox{\includegraphics[width=.32\textwidth]{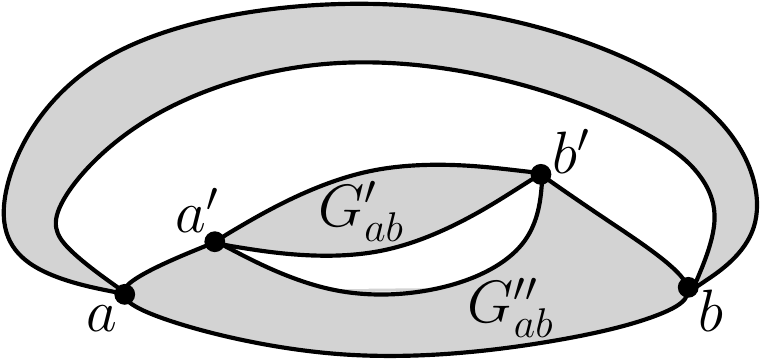}}\\
(a) \hspace{6mm} & (b)  \hspace{6mm} & (c) 
\end{tabular}
\caption{(a) $G_{ab}$ contains more than two non-trivial $\{a',b'\}$-components. (b) $G'_{ab}$ does not contain external vertices of $G_{ab}$. (c) $a'$ and $b'$ both belong to $\tau_{ab}(G_{ab})$.}
\label{fig:structure-proof}
\end{center}
\end{figure}

There are at most two faces $f_i$ of $G_{ab}$, with $i=1,2$, such that both $G'_{ab}$ and $G''_{ab}$ contain vertices different from $a'$ and $b'$ incident to $f_i$. If the outer face of $G_{ab}$ was not one of $f_1$ and $f_2$, as in Fig.~\ref{fig:structure-proof}(b), then one of $G'_{ab}$ and $G''_{ab}$ would be a non-trivial $\{a',b'\}$-component of $G$ that contains no external vertex of $G$ different from $a'$ and $b'$, a contradiction to Property (e) of $(G,u,v,X)$. It follows that both $G'_{ab}$ and $G''_{ab}$ contain external vertices of $G_{ab}$ different from $a'$ and $b'$; also, $a'$ and $b'$ are external vertices of $G_{ab}$. Now assume, for a contradiction, that $a'$ and $b'$ both belong to $\tau_{ab}(G_{ab})$, as in Fig.~\ref{fig:structure-proof}(c) (possibly $a'=a$, or $b'=b$, or both). Then $a$ and $b$ are both contained in the $\{a',b'\}$-component of $G_{ab}$, say $G''_{ab}$, containing $\beta_{ab}(G_{ab})$. It follows that $G'_{ab}$ is a non-trivial $\{a',b'\}$-component of $G$ containing no external vertex of $G$ different from $a'$ and $b'$, a contradiction to Property (e) of $(G,u,v,X)$. Hence, at least one of $a'$ and $b'$ is an internal vertex of $\beta_{ab}(G_{ab})$.
\item Property (f): the vertices in $X_{ab}$ have degree $2$ in $G_{ab}$ and are in $\beta_{ab}(G_{ab})$ since they have degree $2$ in $G$ and are in $\beta_{uv}(G)$. Note that $a,b\notin X$; indeed $G_{ab}$ is biconnected and both $a$ and $b$ have neighbors not in $G_{ab}$, hence $\delta_G(a)=\delta_G(b)=3$.
\end{itemize}

For the induction, we distinguish three cases. 

In the first case $a$ has a unique neighbor $a'$ in $G_{ab}$. Then $a'$ is an internal vertex of $\beta_{uv}(G)$. Since we are not in the base case, $G_{ab}$ is not a simple path with two edges; hence, $\{a',b\}$ is a separation pair of $G$ satisfying the conditions of the lemma. Let $G_{a'b}$ be the $\{a',b\}$-component of $G$ containing $\beta_{a'b}(G)$. Then $G_{ab}$ consists of $G_{a'b}$ together with vertex $a$ and edge $(a,a')$ and induction applies to $G_{a'b}$. If $G_{a'b}$ coincides with $\beta_{a'b}(G)$, then $G_{ab}$ coincides with $\beta_{ab}(G)$, contradicting the fact that we are not in the base case. Hence, $G_{a'b}$ consists of: (i) a path $P'_0=(a',\dots,u_1)$; (ii) for $i=1,\dots,k$ with $k\geq 1$, a biconnected component $G_i$ of $G_{a'b}$ that contains vertices $u_i$ and $v_i$ and such that $(G_i,u_i,v_i,X_i)$ is a well-formed quadruple; (iii) for $i=1,\dots,k-1$, a path $P_i=(v_i,\dots,u_{i+1})$, where $u_{i+1}\neq v_i$; and (iv) a path $P_{k}=(v_{k},\dots,b)$. Then $G_{ab}$ is composed of: (i) path $(a,a')\cup P'_0$; (ii) for $i=1,\dots,k$, the biconnected component $G_i$ of $G_{ab}$; (iii) for $i=1,\dots,k-1$, path $P_i$; and (iv) path $P_{k}$. 

The second case, in which $b$ has a unique neighbor in $G_{ab}$, is symmetric to the first one.

In the third case, the degree of both $a$ and $b$ in $G_{ab}$ is greater than $1$.  Let $G_1$ be the biconnected component of $G_{ab}$ containing $a$. Let $H$ be the subgraph of $G_{ab}$ induced by the vertices with incident edges not in $G_1$. We prove the following claim: $b\notin V(G_1)$, and $H$ and $G_1$ share a single vertex $a'\neq b$, which is an internal vertex of $\beta_{uv}(G)$. 


Assume, for a contradiction, that $b\in V(G_1)$. Then $G_{ab}$ is biconnected. Indeed, if $G_1$ contains a cut-vertex of $G_{ab}$, then this cut-vertex is also a cut-vertex of $G$, since $\{a,b\}$ is a separation pair of $G$ and $G_{ab}$ is an $\{a,b\}$-component of $G$; however, by Property (a) of $(G,u,v,X)$ graph $G$ is biconnected. By the biconnectivity of $G_{ab}$ and the maximality of $G_1$ we have $G_1=G_{ab}$; hence, we are in the base case, a contradiction. 

Every $G_1\cup \{b\}$-bridge of $G_{ab}$ has exactly one attachment in $G_1$ and there is exactly one $G_1\cup \{b\}$-bridge $H$; otherwise, $G_{ab}$ would contain a path not in $G_1$ between two vertices of $G_1$, contradicting the maximality of $G_1$. Denote by $a'$ the only attachment of $H$ in $G_1$. Note that $\delta_{H}(a')=1$, as $\delta_{G_1}(a')\geq 2$ since $G_1$ is biconnected. By the planarity of $G$, we have that $a'$ is incident to the outer face of $G_1$, since $a$ and $b$ are both incident to the outer face of $G$. Since $a'$ is the only attachment of $H$ in $G_1$, it follows that $a'$ is an internal vertex of $\beta_{uv}(G)$. This concludes the proof of the claim.

By the claim and since $G_1$ and $H$ are not single edges, given that the degree of both $a$ and $b$ in $G_{ab}$ is greater than $1$, it follows that $\{a,a'\}$ and $\{a',b\}$ are separation pairs of $G$ satisfying the statement of the lemma, hence induction applies to $G_1$ and $H$. In particular, $(G_1,u_1,v_1,X_1)$ is a well-formed quadruple, with $X_1=X\cap V(G_1)$, $u_1=a$ and $v_1=a'$. Further, $H$ consists of: (i) for $i=1,\dots,k-1$ with $k\geq 2$, a path $P_i=(v_i,\dots,u_{i+1})$ where $u_{i+1}\neq v_i$; note that $P_1=(v_1=a',\dots,u_2)$ satisfies $u_2\neq a'$ since $\delta_H(a')=1$; (ii) for $i=2,\dots,k$, a biconnected component $G_i$ of $H$ containing vertices $u_i$ and $v_i$ (with $v_k=b$) and such that $(G_i,u_i,v_i,X_i)$ is a well-formed quadruple, with $X_i=X\cap V(G_i)$. Then $G_{ab}$ is composed of: (i) a path $P_0=(a)$; (ii) for $i=1,\dots,k$ with $k\geq 1$, a biconnected component $G_i$ that contains vertices $u_i$ and $v_i$ and such that $(G_i,u_i,v_i,X_i)$ is a well-formed quadruple; (iii) for $i=1,\dots,k-1$, a path $P_i=(v_i,\dots,u_{i+1})$, where $u_{i+1}\neq v_i$; and (iv) a path $P_{k}=(b)$. This concludes the proof of the lemma.
\end{proof}


We are now ready to prove Lemma~\ref{le:cubic-2connected}. The proof is by induction on the size of $G$. 

{\bf Base case}: $G$ is a simple cycle. Refer to Fig.~\ref{fig:base-case}. If $u$ and $v$ were not adjacent, then $\{u,v\}$ would be a separation pair none of whose vertices is internal to $\beta_{uv}(G)$, contradicting Property (e) of $(G,u,v,X)$. Thus, edge $(u,v)$ exists and coincides with $\tau_{uv}(G)$ by Property (d). We now construct a proper good curve $\lambda$. Curve $\lambda$ starts at $u$; it then passes through all the vertices in $V(G)-(X\cup\{v\})$ in the order in which they appear along $\beta_{uv}(G)$ from $u$ to $v$; in particular, if two vertices in $V(G)-(X\cup\{v\})$ are consecutive in $\beta_{uv}(G)$, then $\lambda$ contains the edge between them. If the neighbor $v'$ of $v$ in $\beta_{uv}(G)$ is not in $X$, then $\lambda$ ends at $v'$, otherwise $\lambda$ ends at a point $z$ in the interior of edge $(v,v')$. Charge $v$ to $u$ and note that $v$ is the only vertex in $V(G)-X$ that is not on $\lambda$. It is easy to see that $\lambda$ is a proper good curve satisfying Properties (1)--(6).

\begin{figure}[htb]
\begin{center}
\mbox{\includegraphics[width=.3\textwidth]{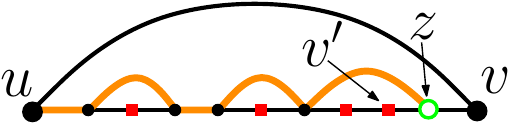}} \caption{Base case for the proof of Lemma~\ref{le:cubic-2connected}.}
\label{fig:base-case}
\end{center}
\end{figure}


Next we describe the inductive cases. In the description of each inductive case, we implicitly assume that  none of the previously described cases applies.  

{\bf Case 1}: edge $(u,v)$ exists. Refer to Fig.~\ref{fig:cubic-case-1}. By Property (d) of $(G,u,v,X)$ edge $(u,v)$ coincides with $\tau_{uv}(G)$. By Property (c), vertex $v$ has a unique neighbor $v'$. Since $G$ is not a simple cycle with length three, $\{u,v'\}$ is a separation pair of $G$ to which Lemma~\ref{le:separation-bottom} applies. If the $\{u,v'\}$-component of $G$ containing $\beta_{uv'}(G)$ coincided with $\beta_{uv'}(G)$, then $G$ would be a simple cycle, a contradiction to the fact that we are not in the base case. Hence, the graph $G'$ obtained from $G$ by removing edge $(u,v)$ consists of: (i) a path $P_0=(u,\dots,u_1)$; (ii) for $i=1,\dots,k$ with $k\geq 1$, a biconnected component $G_i$ of $G'$ that contains vertices $u_i$ and $v_i$ and such that $(G_i,u_i,v_i,X_i)$ is a well-formed quadruple, where $X_i=X\cap V(G_i)$; (iii) for $i=1,\dots,k-1$, a path $P_i=(v_i,\dots,u_{i+1})$, where $u_{i+1}\neq v_i$; and (iv) a path $P_{k}=(v_{k},\dots,v)$. Inductively compute a curve $\lambda_i$ satisfying the properties of Lemma~\ref{le:cubic-2connected} for each quadruple $(G_i,u_i,v_i,X_i)$. We construct a proper good curve $\lambda$ for $(G,u,v,X)$ as follows. 

\begin{figure}[htb]
\begin{center}
\mbox{\includegraphics[width=.7\textwidth]{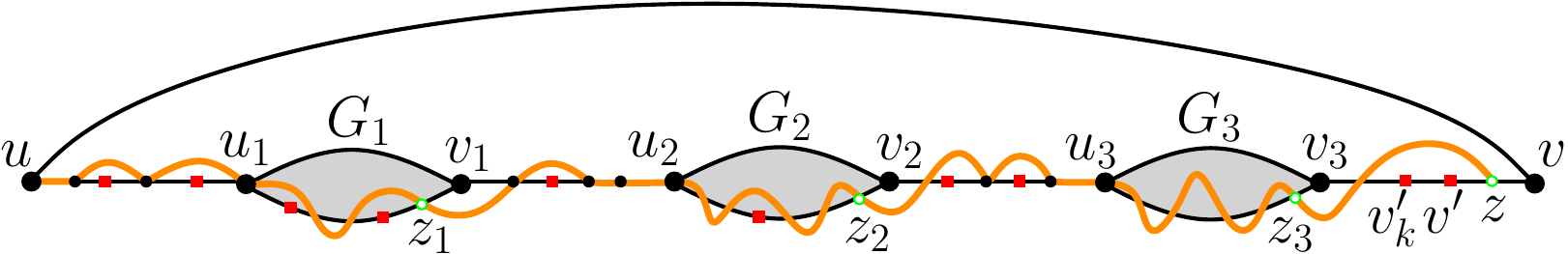}} \caption{Case 1 of the proof of Lemma~\ref{le:cubic-2connected} with $k=3$.}
\label{fig:cubic-case-1}
\end{center}
\end{figure}


\begin{itemize}
\item Curve $\lambda$ starts at $u$. 
\item It then passes through all the vertices in $V(P_0)\setminus X$ in the order as they appear along $\beta_{uv}(G)$ from $u$ to $u_1$; note that $u_1\notin X$, since $\delta_G(u_1)=3$, hence $\lambda$ passes through $u_1$; this part of $\lambda$ lies in the internal face of $G$ incident to edge $(u,v)$. 
\item Suppose that $\lambda$ has been constructed up to a vertex $u_i$, for some $1\leq i\leq k$. Then $\lambda$ contains $\lambda_i$, which terminates at a point $z_i$ on $\beta_{u_iv_i}(G_i)$. 
\item Suppose that $\lambda$ has been constructed up to a point $z_i$ on $\beta_{u_iv_i}(G_i)$, for some $1\leq i\leq k-1$. Then $\lambda$ continues with a curve in the outer face of $G$ from $z_i$ to the neighbor $v'_i$ of $v_i$ in $P_i$ (if $v'_i\notin X$, as with $i=1$ in Fig.~\ref{fig:cubic-case-1}) or from $z_i$ to a point in the interior of edge $(v_i,v'_i)$ (if $v'_i\in X$, as with $i=2$ in Fig.~\ref{fig:cubic-case-1}). 
\item Suppose that $\lambda$ has been constructed up to a point on edge $(v_i,v'_i)$ (possibly coinciding with $v'_i$), for some $1\leq i\leq k-1$. Then $\lambda$ passes through all the vertices in $V(P_i)\setminus(X\cup \{v_i\})$ in the order as they appear along $\beta_{uv}(G)$ from $v_i$ to $u_{i+1}$; note that $u_{i+1}\notin X$, since $\delta_G(u_{i+1})=3$, hence $\lambda$ passes through $u_{i+1}$; this part of $\lambda$ lies in the internal face of $G$ incident to edge $(u,v)$.
\item Finally, suppose that $\lambda$ has been constructed up to a point $z_k$ on $\beta_{u_kv_k}(G_k)$. If the neighbor $v'_k$ of $v_k$ in $P_k$ is $v$, then $\lambda$ terminates at $z_k$. Otherwise, $\lambda$ continues with a curve in the outer face of $G$ from $z_k$ to $v'_k$ (if $v'_k\notin X$) or from $z_k$ to a point in the interior of edge $(v_k,v'_k)$ (if $v'_k\in X$). Then $\lambda$ passes through all the vertices in $V(P_k)\setminus(X\cup \{v_k,v\})$ in the order as they appear along $\beta_{uv}(G)$ from $v_k$ to $v$. If $v'\in X$, then $\lambda$ terminates at a point $z$ along edge $(v',v)$, otherwise $\lambda$ terminates at $v'$. 
\end{itemize}

Curve $\lambda$ satisfies Properties (1)--(5) of Lemma~\ref{le:cubic-2connected}. In particular, the part of $\lambda$ from $z_i$ to a point on edge $(v_i,v'_i)$ can be drawn without causing self-intersections because $\lambda_i$ satisfies Properties (2), (3), and (5) by induction; in fact, these properties ensure that $z_i$ and $v'_i$ are both incident to $R_{G,\lambda_i}$. For $i=1,\dots,k$, the charge of the vertices in $(N_{\lambda}\cap V(G_i))$ to the vertices in $L_{\lambda}\cap V(G_i)$ is determined inductively, thus each vertex in $L_{\lambda}\cap V(G_i)$ is charged with at most three vertices; charge $v$ to $u$ and observe that Property (6) is satisfied by the constructed charging scheme.  

If Case~1 does not apply, then consider the graph $G'=G-\{v\}$. Since $\{u,v\}$ is not a separation pair of $G$, then $u$ is not a cut-vertex of $G'$. Let $H$ be the biconnected component of $G'$ containing $u$. We have the following claim.

\begin{claimx} \label{cl:structure-h} 
Graph $G$ has two $H\cup \{v\}$-bridges $B_1$ and $B_2$; further, each of $B_1$ and $B_2$ has two attachments, one of which is $v$; finally, one of $B_1$ and $B_2$ is an edge of $\tau_{uv}(G)$. 
\end{claimx}

\begin{proof}
First, each $H\cup \{v\}$-bridge $B_i$ of $G$ has at most one attachment $y_i$ in $H$, as otherwise $B_i$ would contain a path (not passing through $v$) between two vertices of $H$, and $H$ would not be maximal. 

Second, if $B_i$ had no attachment in $H$, then $v$ would be a cut-vertex of $G$, whereas $G$ is biconnected. Also, if $v$ was not an attachment of $B_i$, then $y_i$ would be a cut-vertex of $G$, whereas $G$ is biconnected. Hence, $B_i$ has two attachments, namely $v$ and $y_i$. Further, if there was a single $H\cup \{v\}$-bridge $B_i$, then $y_i$ would be a cut-vertex of $G$, whereas $G$ is biconnected. This and $\delta_G(v)=2$ imply  that $G$ has two $H\cup \{v\}$-bridges $B_1$ and $B_2$. 

Finally, one of $y_1$ and $y_2$, say $y_1$, belongs to $\tau_{uv}(G)$, while the other one, say $y_2$, belongs to $\beta_{uv}(G)$. Hence, if $B_1$ was not a trivial $H\cup \{v\}$-bridge, then $\{y_1,v\}$ would be a separation pair none of whose vertices is internal to $\beta_{uv}(G)$, whereas $(G,u,v,X)$ is a well-formed quadruple. This concludes the proof of the claim.
\end{proof}

By Claim~\ref{cl:structure-h} graph $G$ is composed of three subgraphs: a biconnected graph $H$, an edge $B_1=(y_1,v)$, and a graph $B_2$, where $H$ and $B_1$ share vertex $y_1$, $H$ and $B_2$ share vertex $y_2$, and $B_1$ and $B_2$ share vertex $v$. Before proceeding with the case distinction, we argue about the structure of $H$. Let $X'=\{y_2\}\cup (X\cap V(H))$. We have the following.

\begin{claimx} \label{cl:H-well-formed}
$(H,u,y_1,X')$ is a well-formed quadruple. 
\end{claimx} 

\begin{proof}
Properties (a)--(c) are trivially satisfied by $(H,u,y_1,X')$. Concerning Property (d), if edge $(u,y_1)$ exists, then it is either $\tau_{uy_1}(H)$ or $\beta_{uy_1}(H)$, since $\delta_H(u)=2$. However, $(u,y_1)\neq \beta_{uy_1}(H)$, since $y_2\in \beta_{uy_1}(H)$ and $y_2\neq u, y_1$. 

Next, we discuss Property (e). Consider any separation pair $\{a,b\}$ of $H$. First, if $a$ was not an external vertex of $H$, then $\{a,b\}$ would also be a separation pair of $G$ such that $a$ is not an external vertex of $G$; this would contradict Property (e) of $(G,u,v,X)$. Second, if both $a$ and $b$ were in $\tau_{uy_1}(H)$, then $\{a,b\}$ would be a separation pair of $G$ whose vertices are both in $\tau_{uv}(G)$, given that $\tau_{uy_1}(H)\subset \tau_{uv}(G)$; again, this would contradict Property (e) of $(G,u,v,X)$. Third, if an $\{a,b\}$-component $H_{ab}$ of $H$ contained no external vertex of $H$ different from $a$ and $b$, then $H_{ab}$ would also be an $\{a,b\}$-component of $G$ containing no external vertex of $G$ different from $a$ and $b$, again contradicting Property (e) of $(G,u,v,X)$.

Finally, we deal with Property (f). The vertices in $X\cap X'$ have degree $2$ in $H$ since they have degree $2$ in $G$ and are internal to $\beta_{uy_1}(H)$ since they are internal to $\beta_{uv}(G)$. Further, we have that $\delta_H(y_2)=2$ since $H$ is biconnected, since $\delta_H(y_2)<\delta_G(y_2)$ (given that $y_2$ has a neighbor in $B_2$ not in $H$), and since $\delta_G(y_2)\leq 3$. Also, $y_2$ is an internal vertex of $\beta_{uy_1}(H)$, since it is an internal vertex of $\beta_{uv}(G)$ and is in $H$. This concludes the proof of the claim. 
\end{proof}



{\bf Case 2}: $B_2$ contains a vertex not in $X\cup\{v,y_2\}$. Refer to Fig.~\ref{fig:cubic-case-2}. Curve $\lambda$ will be composed of three curves $\lambda_1$, $\lambda_2$, and $\lambda_3$. 

\begin{figure}[htb]
\begin{center}
\mbox{\includegraphics[width=.6\textwidth]{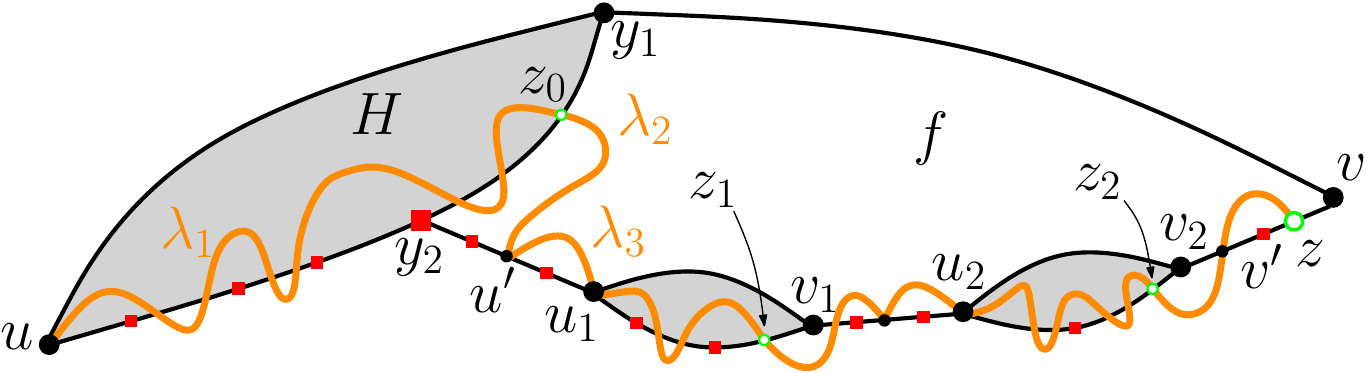}} \caption{Case 2 of the proof of Lemma~\ref{le:cubic-2connected}.}
\label{fig:cubic-case-2}
\end{center}
\end{figure}

Curve $\lambda$ starts at $u$. By Claim~\ref{cl:H-well-formed}, a curve $\lambda_1$ satisfying the properties of Lemma~\ref{le:cubic-2connected} can be inductively constructed for $(H,u,y_1,X')$. Notice that $y_2\in X'$, thus $\lambda_1$ terminates at a point $z_0$ in $\beta_{y_2y_1}(H)$, by Property (2) of $\lambda_1$.  

Curve $\lambda_2$ lies in the internal face $f$ of $G$ incident to edge $(y_1,v)$ and connects $z_0$ with a vertex $u'$ in $B_2$ determined as follows. Traverse $\beta_{uv}(G)$ from $y_2$ to $v$ and let $u'\neq y_2$ be the first encountered vertex not in $X$. By Property (f) of $(G,u,v,X)$, every vertex in $X\cap V(B_2)$ has degree $2$ in $G$ and in $B_2$; also, $\delta_{B_2}(y_2)=\delta_{B_2}(v)=1$. If all the internal vertices of $\beta_{y_2 v}(G)$ belong to $X$, then $B_2$ is a path whose internal vertices are in $X$, a contradiction to the hypothesis of Case 2. Hence, $u'\neq v$, $\beta_{y_2u'}(G)$ is induced in $B_2$, $u'$ is incident to $f$, and the interior of $\lambda_2$ crosses no edge of $G$. It is vital here that $\lambda_1$ satisfies Properties (3)--(5), ensuring that $y_2$ is not on $\lambda_1$ and that the edge incident to $y_2$ in $B_2$ is in $R_{G,\lambda_1}$. Thus, if such an edge is $(y_2,u')$, still $\lambda$ intersects it only once.
 
Curve $\lambda_3$ connects $u'$ with a point $z\neq y_2,v$ on $\beta_{y_2 v}(G)$. Note that $\{y_2,v\}$ is a separation pair of $G$, since by hypothesis $B_2$ is not an edge; further, $y_2$ and $v$ both belong to $\beta_{uv}(G)$. Hence Lemma~\ref{le:separation-bottom} applies and curve $\lambda_3$ is constructed as in Case~1. 

Curve $\lambda$ satisfies Properties (1)--(5) of Lemma~\ref{le:cubic-2connected}. We determine inductively the charge of the vertices in $(N_{\lambda}\cap V(H))-\{y_2\}$ to the vertices in $L_{\lambda}\cap V(H)$, and the charge of the vertices in $N_{\lambda}$ in each biconnected component $G_i$ of $B_2$ to the vertices in $L_{\lambda}\cap V(G_i)$. The only vertices in $N_{\lambda}$ that have not yet been charged to vertices in $L_{\lambda}$ are $y_2$ and $v$; charge them to $u'$. Then $u$ is charged with at most $1$ vertex of $H$; every vertex in $L_{\lambda}-\{u,u'\}$ is charged with at most $3$ vertices if it is in $H$ or in a biconnected component of $B_2$, or with no vertex otherwise; finally, $u'$ is charged with $y_2$, $v$, and with no other vertex if $\delta_{G}(u')=2$ or with at most $1$ other vertex if $\delta_{G}(u')=3$; indeed, in the latter case $u'=u_1$ is such that induction is applied on a quadruple $(G_1,u_1,v_1,X_1)$. Thus, Property (6) is satisfied by the constructed charging scheme.  

If Case~2 does not apply, then $B_2$ is a path between $y_2$ and $v$ whose internal vertices are in $X$. In order to proceed with the case distinction, we explore the structure of $H$.


{\bf Case 3}: edge $(u,y_1)$ exists. By Claim~\ref{cl:H-well-formed}, $(H,u,y_1,X')$ is a well-formed quadruple, thus by Property (d) edge $(u,y_1)$ coincides with $\tau_{uy_1}(H)$. Let $y'$ be the unique neighbor of $y_1$ in $\beta_{uy_1}(H)$.

\begin{figure}[htb]
\begin{center}
\begin{tabular}{c c}
\mbox{\includegraphics[width=.22\textwidth]{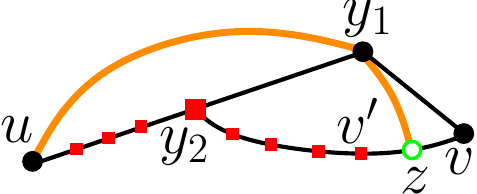}}\hspace{3mm} &
\mbox{\includegraphics[width=.55\textwidth]{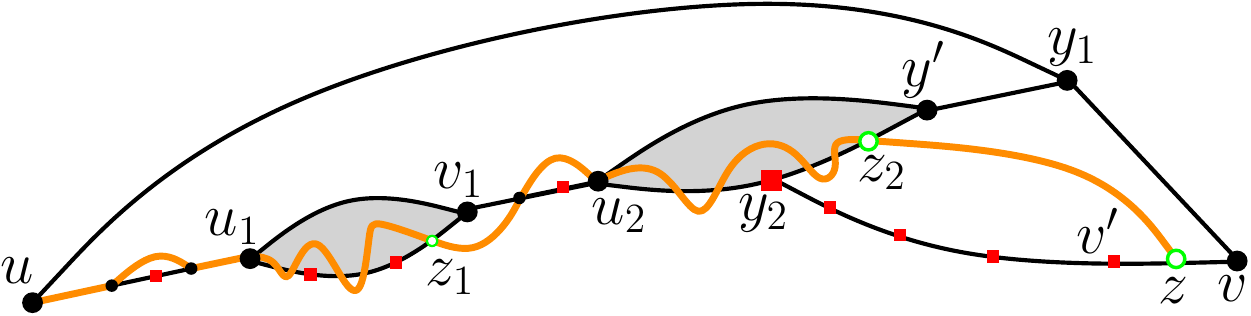}}\\
(a) \hspace{3mm} & (b) 
\end{tabular}
\caption{Case 3 of the proof of Lemma~\ref{le:cubic-2connected}. (a) Every vertex of $H$ different from $u$ and $y_1$ is in $X'$. (b) $H$ contains a vertex not in $X'\cup \{u,y_1\}$.}
\label{fig:cubic-case-3}
\end{center}
\end{figure}

If every vertex of $H$ different from $u$ and $y_1$ is in $X'$ (as in Fig.~\ref{fig:cubic-case-3}(a)), then $\lambda$ consists of edge $(u,y_1)$ together with a curve from $y_1$ to a point $z$ along edge $(v,v')$; the latter curve lies in the internal face of $G$ incident to edge $(v,y_1)$. Charge $y_2$ and $v$ to $y_1$ and note that $\lambda$ satisfies Properties (1)--(6) required by Lemma~\ref{le:cubic-2connected}. 

If $H$ contains a vertex not in $X'\cup \{u,y_1\}$  (as in Fig.~\ref{fig:cubic-case-3}(b)), then $H$ contains at least $4$ vertices; also, $u$ and $y'$ belong to $\beta_{uy_1}(H)$. Thus, Lemma~\ref{le:separation-bottom} applies to separation pair $\{u,y'\}$ of $H$ and a curve $\lambda_1$ can be constructed that connects $u$ with a point $z_k\neq y_1$ on $\beta_{y_2y_1}(H)$ as in Case~1. Curve $\lambda$ consists of $\lambda_1$ and of a curve $\lambda_2$ lying in the internal face of $G$ incident to edge $(v,y_1)$ and connecting $z_k$ with a point $z$ along edge $(v,v')$. Curve $\lambda$ satisfies Properties (1)--(5) of Lemma~\ref{le:cubic-2connected}. We determine inductively the charge of the vertices in $N_{\lambda}-\{y_2\}$ in each biconnected component $G_i$ of the graph obtained from $H$ by removing edge $(u,y_1)$ to the vertices in $L_{\lambda}\cap V(G_i)$. We charge $v$ to $u$, and $y_1$ and $y_2$ to the first vertex $u'\neq u$ not in $X'$ encountered when traversing $\beta_{uy_1}(H)$ from $u$ to $y_1$. That $u'$ exists, that $u'\neq y_1$, and that $u'\in L_{\lambda}$ can be proved as in Case~2 by the assumption that $H$ contains a vertex not in $X'\cup\{u,y_1\}$; then either zero or one vertex has been charged to $u'$ so far, depending on whether $\delta_G(u')=2$ or $\delta_G(u')=3$, respectively, and Property (6) is satisfied by the constructed charging scheme.

If Case~3 does not apply, consider the graph $H'=H-\{y_1\}$. Since we are not in Case~3, $(u,y_1)$ is not an edge of $H$; also, by Claim~\ref{cl:H-well-formed} and Property (e) of $(H,u,y_1,X')$, $\{u,y_1\}$ is not a separation pair of $H$. It follows that $u$ is not a cut-vertex of $H'$. Let $K$ be the biconnected component of $H'$ containing $u$. Analogously as in Claim~\ref{cl:structure-h}, it can be proved that $H$ has two $K\cup \{y_1\}$-bridges $D_1$ and $D_2$, that $D_1$ is a trivial $K\cup \{y_1\}$-bridge $(w_1,y_1)$ which is an edge of $\tau_{uy_1}(H)$ and that $D_2$ has two attachments $w_2$ and $y_1$. We further distinguish the cases in which $y_2$ does or does not belong to $K$. 


{\bf Case 4: $y_2 \in K$}. Refer to Fig.~\ref{fig:cubic-case-4}. Vertices $y_2$ and $w_2$ are distinct. Indeed, if they were the same vertex, then $\delta_G(y_2)\geq 4$, as $y_2$ would have at least two neighbors in $K$, since $K$ is biconnected, and one neighbor in each of $B_2$ and $D_2$; however, this would contradict the fact that $G$ is a subcubic graph. Since $w_1,y_1 \in \tau_{uv}(G)$ and $y_2\in \beta_{uv}(G)$, vertices $u,y_2,w_2,w_1$ come in this order along $\beta_{uw_1}(K)$; it follows that $D_2$ is a trivial $K\cup \{y_1\}$-bridge, as otherwise $\{y_1,w_2\}$ would be a separation pair of $G$ one of whose vertices is internal to $G$, while $(G,u,v,X)$ is a well-formed quadruple.

\begin{figure}[htb]
\begin{center}
\mbox{\includegraphics[width=.4\textwidth]{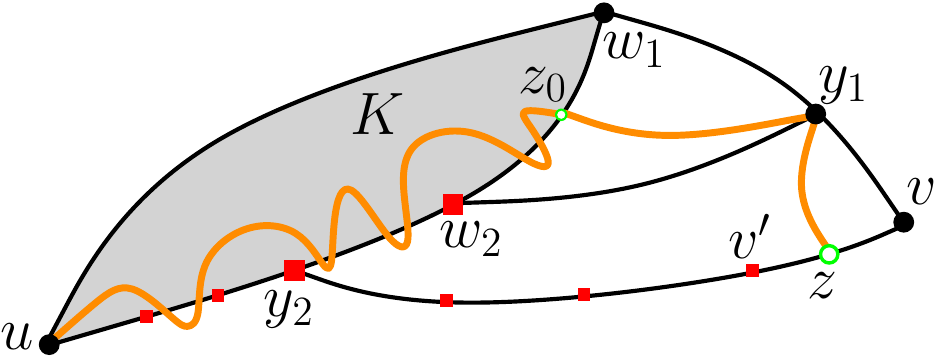}}
\caption{Case 4 of the proof of Lemma~\ref{le:cubic-2connected}.}
\label{fig:cubic-case-4}
\end{center}
\end{figure}

Let $X''=(X\cap V(K))\cup \{y_2,w_2\}$. Analogously as in Claim~\ref{cl:H-well-formed}, it can be proved that $(K,u,w_1,X'')$ is a well-formed quadruple.  By induction, a curve $\lambda_1$ can be constructed satisfying the properties of Lemma~\ref{le:cubic-2connected} for $(K,u,w_1,X'')$. In particular, $\lambda_1$ starts at $u$ and ends at a point $z_0\neq w_1$ in $\beta_{w_2w_1}(K)$. Curve $\lambda$ consists of $\lambda_1$, of a curve $\lambda_2$ from $z_0$ to $y_1$ lying in the internal face of $G$ incident to edge $(w_1,y_1)$, and of a curve $\lambda_3$ from $y_1$ to a point $z$ along edge $(v,v')$ lying in the internal face of $G$ incident to edge $(y_1,v)$. Curve $\lambda$ satisfies Properties (1)--(5) of Lemma~\ref{le:cubic-2connected}. Property (6) is satisfied by charging the vertices in $(N_{\lambda}\cap V(K))-\{y_2,w_2\}$ to the vertices in $L_{\lambda}\cap V(K)$ as computed by induction, and by charging $v$, $y_2$, and $w_2$ to $y_1$.


{\bf Case 5: $y_2 \notin K$}. Let $X''=\{w_2\}\cup (X\cap V(K))$. It can be proved as in Claim~\ref{cl:H-well-formed} that $(K,u,w_1,X'')$ is a well-formed quadruple.  By induction, a curve $\lambda_1$ can be constructed satisfying the properties of Lemma~\ref{le:cubic-2connected} for $(K,u,w_1,X'')$. In particular, $\lambda_1$ starts at $u$ and ends at a point $z_0\neq w_1$ in $\beta_{w_2w_1}(K)$. Curve $\lambda_1$ is the first part of $\lambda$. 

\begin{figure}[htb]
\begin{center}
\begin{tabular}{c c}
\mbox{\includegraphics[width=.4\textwidth]{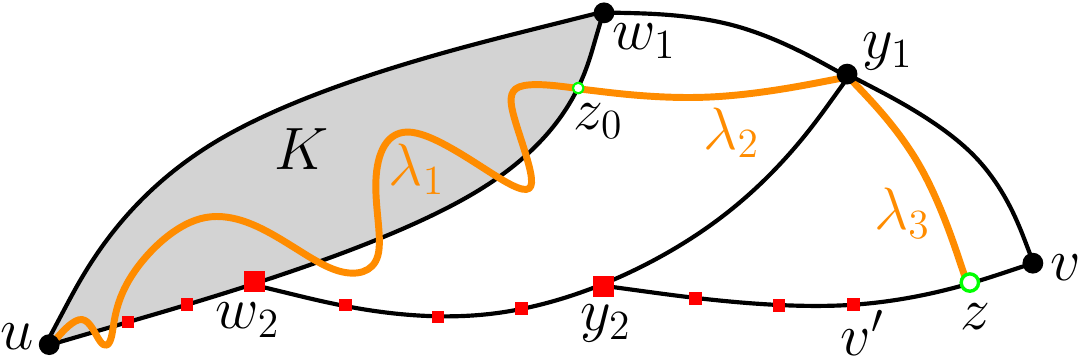}}\hspace{1mm} &
\mbox{\includegraphics[width=.5\textwidth]{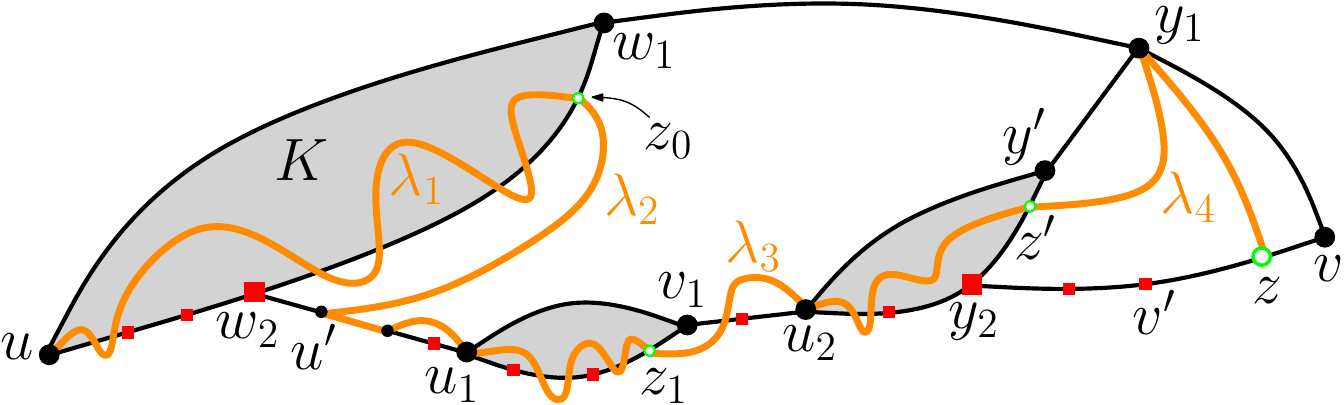}}\\
(a) \hspace{1mm} & (b) 
\end{tabular}
\caption{Case 5 of the proof of Lemma~\ref{le:cubic-2connected}. (a) Every vertex of $D_2$ different from $w_2$ and $y_1$ is in $X'$. (b) $D_2$ contains a vertex not in $X'\cup \{y_1,w_2\}$.}
\label{fig:cubic-case-5}
\end{center}
\end{figure}

If every vertex of $D_2$ different from $w_2$ and $y_1$ is in $X'$, as in Fig.~\ref{fig:cubic-case-5}(a), then $\lambda$ continues with a curve $\lambda_2$ that connects $z_0$ with $y_1$ ($\lambda_2$ lies in the internal face of $G$ incident to edge $(w_1,y_1)$) and with a curve $\lambda_3$ that connects $y_1$ with a point $z$ along edge $(v',v)$ ($\lambda_3$ lies in the internal face of $G$ incident to edge $(y_1,v)$). 

If $D_2$ contains a vertex not in $X'\cup\{y_1,w_2\}$, as in Fig.~\ref{fig:cubic-case-5}(b), then, similarly to Case~2, $\lambda$ continues with a curve $\lambda_2$ that connects $z_0$ with the first vertex $u'\neq w_2$ not in $X'$ encountered while traversing $\beta_{w_2y_1}(H)$ from $w_2$ to $y_1$; curve $\lambda_2$ lies in the internal face of $G$ incident to edge $(w_1,y_1)$. That $u'$ exists, that $u'\neq y_1$, and that $u'\in L_{\lambda}$ can be proved as in Case~2 by the assumption that $D_2$ contains a vertex not in $X'\cup\{y_1,w_2\}$. Then $\lambda$ continues with a curve $\lambda_3$ that connects $u'$ with a point $z'$ in $\beta_{y_2y_1}(H)$; as in Case~2, $\{w_2,y_1\}$ is a separation pair of $H$, hence Lemma~\ref{le:separation-bottom} applies and curve $\lambda_3$ is constructed as in Case~1. Finally, if $z'$ is not a point internal to edge $(y',y_1)$, curve $\lambda$ contains a curve $\lambda_4$ that connects $z'$ with $y_1$, and then $y_1$ with a point $z$ on edge $(v,v')$; curve $\lambda_4$ lies in the internal face of $G$ incident to edge $(y_1,v)$. Otherwise, we redraw the last part of $\lambda_3$ so that it terminates at $y_1$ rather than at $z'$; we then let $\lambda_4$ connect $y_1$ with a point $z$ on edge $(v,v')$ in the internal face of $G$ incident to edge $(y_1,v)$.

Curve $\lambda$ satisfies Properties (1)--(5) of Lemma~\ref{le:cubic-2connected}. We determine inductively the charge of the vertices in $(N_{\lambda}\cap V(K))-\{w_2\}$ to the vertices in $L_{\lambda}\cap V(K)$, as well as the charge of the vertices in $N_{\lambda}-\{y_2\}$ in each biconnected component $G_i$ of $D_2$, if any, to the vertices in $L_{\lambda}\cap V(G_i)$. Charge $v$, $y_2$, and $w_2$ to $y_1$. Property (6) is satisfied by the constructed charging scheme. This concludes the proof of Lemma~\ref{le:cubic-2connected}. 



We now apply Lemma~\ref{le:cubic-2connected} to prove Theorem~\ref{th:cubic}. Let $G$ be any triconnected cubic plane graph. Let $G'$ be the plane graph obtained from $G$ by removing any edge $(u,v)$ incident to the outer face of $G$, where $u$ is encountered right before $v$ when walking in clockwise direction along the outer face of $G$. Let $X'=\emptyset$. We have the following.

\begin{lemma} \label{le:trico-to-lemma}
$(G',u,v,X')$ is a well-formed quadruple.
\end{lemma}

\begin{proof}
Concerning Property (a) $G'$ is a subcubic plane graph since $G$ is. Also, $G'$ is biconnected, since $G$ is triconnected. Concerning Property (b), vertices $u$ and $v$ are external vertices of $G'$ since they are external vertices of $G$. Concerning Property (c), $\delta_{G'}(u)=\delta_{G'}(v)=2$ since $\delta_{G}(u)=\delta_{G}(v)=3$. Properties (d) and (f) are trivially satisfied since edge $(u,v)$ does not belong to $G'$ and since $X'=\emptyset$, respectively. 

We now prove Property (e). Consider any separation pair $\{a,b\}$ of $G'$. If $G'$ had at least $3$ non-trivial $\{a,b\}$-components, then $G$ would have at least $2$ non-trivial $\{a,b\}$-components, whereas it is triconnected. Hence, $G'$ has $2$ non-trivial $\{a,b\}$-components $H$ and $H'$. Vertices $u$ and $v$ are not in the same non-trivial $\{a,b\}$-component of $G'$, as otherwise $G$ would not be triconnected. This implies that $\{a,b\}\cap \{u,v\}=\emptyset$. Both $H$ and $H'$ contain external vertices of $G'$ (in fact $u$ and $v$). It follows that $a$ and $b$ are both external vertices of $G'$. Hence, vertices $u$, $a$, $v$, and $b$ come in this order along the boundary of the outer face of $G'$, thus one of $a$ and $b$ is internal to $\tau_{uv}(G')$, while the other one is internal to $\beta_{uv}(G')$. This concludes the proof of the lemma. 
\end{proof}


It follows by Lemma~\ref{le:trico-to-lemma} that a proper good curve $\lambda$ can be constructed satisfying the properties of Lemma~\ref{le:cubic-2connected}. Insert the edge $(u,v)$ in the outer face of $G'$, restoring the plane embedding of $G$. By Properties (1)--(5) of $\lambda$ this insertion can be accomplished so that $(u,v)$ does not intersect $\lambda$ other than at $u$, hence $\lambda$ remains proper and good. In particular, the end-points $u$ and $z$ of $\lambda$ both belong to $\beta_{uv}(G')$, while the insertion of $(u,v)$ only prevents the internal vertices of $\tau_{uv}(G')$ from being incident to $R_{G,\lambda}$. By Property (6) of $\lambda$ with $X'=\emptyset$, each vertex in $N_{\lambda}$ is charged to a vertex in $L_{\lambda}$, and each vertex in $L_{\lambda}$ is charged with at most three vertices in $N_{\lambda}$. Thus, $\lambda$ is a proper good curve passing through $\lceil \frac{n}{4}\rceil$ vertices of $G$. This concludes the proof of Theorem~\ref{th:cubic}.



\section{Planar Graphs with Large Tree-width} \label{le:treewidth}

In this section we prove the following theorem.

\begin{theorem} \label{th:large-treewidth}
Let $G$ be a planar graph and $k$ be its tree-width. There exists a \cfsl\ of $G$ with $\Omega(k^2)$ collinear vertices.
\end{theorem}


Let $G$ be a planar graph with tree-width $k$. We assume that $G$ is connected; indeed, if it is not, edges can be added to it in order to make it connected. This augmentation does not decrease the tree-width of $G$; further, the added edges can be removed once a \cfsl\ of the augmented graph with $\Omega(k^2)$ collinear vertices has been constructed. In order to prove that $G$ admits a \cfsl\ with $\Omega(k^2)$ collinear vertices we exploit Theorem~\ref{th:topology}, as well as a result of Robertson, Seymour and Thomas~\cite{rst-qepg-94}, which asserts that $G$ contains a $g\times g$ grid $H$ as a minor, where $g$ is at least $(k+4)/6$.

Denote by $v_{i,j}$ the vertices of $H$, with $1\leq i,j\leq g$, where $v_{i,j}$ and $v_{i',j'}$ are adjacent in $H$ if and only if $|i-i'|+|j-j'|=1$. Denote by $G_{i,j}$ the connected subgraph of $G$ represented by $v_{i,j}$ in $H$. By the planarity of $G$, every edge of $G$ that is incident to a vertex in $G_{i,j}$, for some $2\leq i,j \leq g-1$, has its other end-vertex in a graph $G_{i',j'}$ such that $|i-i'|\leq 1$ and $|j-j'|\leq 1$. (The previous statement might not be true for an edge that is incident to a vertex in $G_{i,j}$ with $i=1$, $i=g$, $j=1$, or $j=g$.) 

Refer to Fig.~\ref{fig:cell}(a). For every edge $(v_{i,j},v_{i+1,j})$ of $H$, arbitrarily choose an edge $e_{i,j}$ connecting a vertex in $G_{i,j}$ and a vertex in $G_{i+1,j}$ as the {\em reference edge} for the edge $(v_{i,j},v_{i+1,j})$ of $H$. Such an edge exists since $H$ is a minor of $G$. Reference edges $e'_{i,j}$ for the edges $(v_{i,j},v_{i,j+1})$ of $H$ are defined analogously. 

\begin{figure}[htb]
\begin{center}
\begin{tabular}{c c}
\mbox{\includegraphics[width=.45\textwidth]{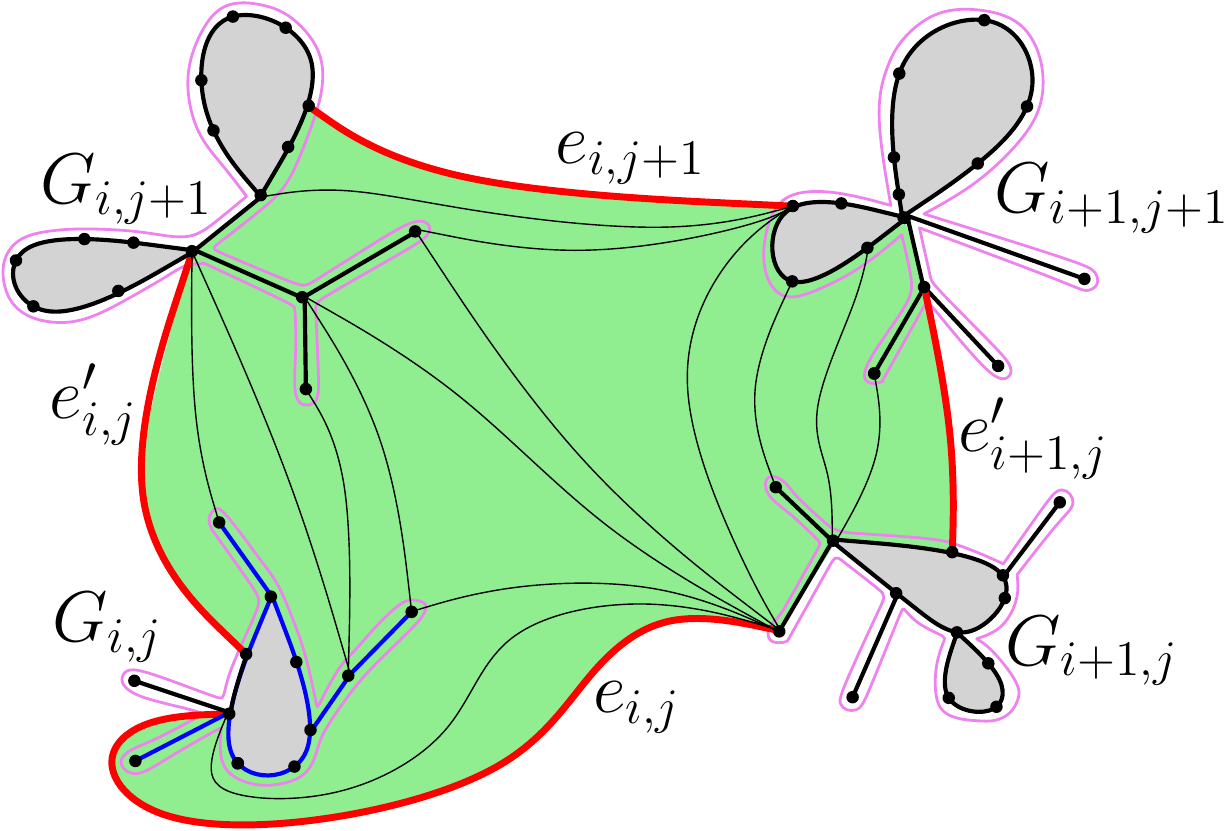}} \hspace{3mm} &
\mbox{\includegraphics[width=.3\textwidth]{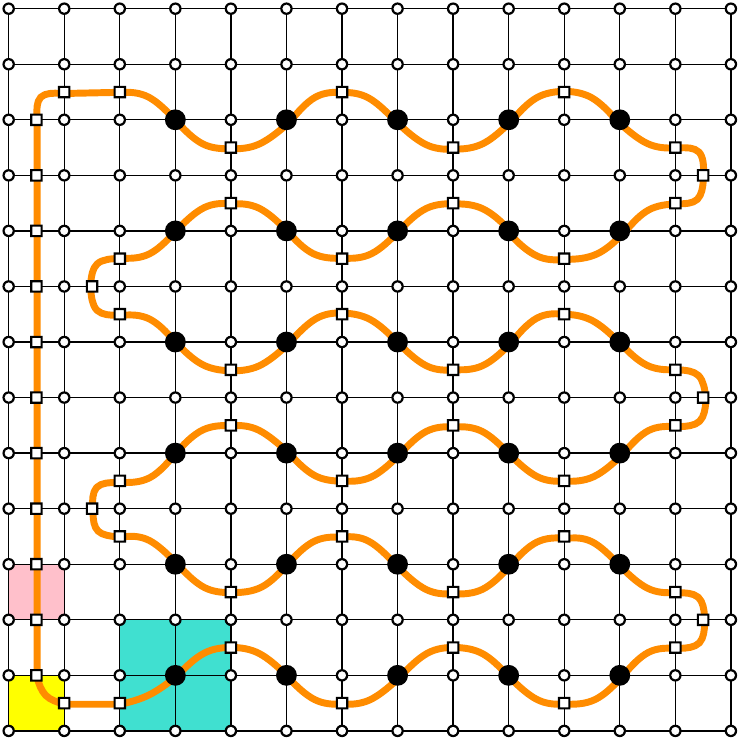}}\\
(a) \hspace{3mm} & (b)
\end{tabular}
\caption{(a) Cells, boundaries, and references edges. Cell $C_{i,j}$ is green. Graphs $G_{i,j}$, $G_{i+1,j}$, $G_{i,j+1}$, and $G_{i+1,j+1}$ are surrounded by violet curves; their interior is gray. The references edges are red and thick. The right-top boundary of $G_{i,j}$ is blue. (b) Construction of $\lambda$ (represented as a thick orange line). Large disks represent graphs $G_{i,j}$ such that $\lambda$ passes through vertices of $G_{i,j}$. Small circles represent graphs $G_{i,j}$ such that $\lambda$ does not pass through any vertex of $G_{i,j}$. White squares represent intersections between $\lambda$ and reference edges.}
\label{fig:cell}
\end{center}
\end{figure}

For every pair of indices $1\leq i,j\leq g-1$, we call {\em right-top} boundary of $G_{i,j}$ the walk that starts at the end-vertex of $e'_{i,j}$ in $G_{i,j}$, traverses the boundary of the outer face of $G_{i,j}$ in clockwise direction and ends at the end-vertex of $e_{i,j}$ in $G_{i,j}$. The {\em right-bottom} boundary of $G_{i,j}$ (for every $1\leq i\leq g-1$ and $2\leq j\leq g$), the {\em left-top} boundary of $G_{i,j}$ (for every $2\leq i\leq g$ and $1\leq j\leq g-1$), and the {\em left-bottom} boundary of $G_{i,j}$ (for every $2\leq i,j\leq g$) are defined analogously.


For each $1\leq i,j\leq g-1$, we define the {\em cell} $C_{i,j}$ as the bounded closed region of the plane that is delimited by (in clockwise order along the boundary of the region): the right-top boundary of $G_{i,j}$, edge $e'_{i,j}$, the right-bottom boundary of $G_{i,j+1}$, edge $e_{i,j+1}$, the left-bottom boundary of $G_{i+1,j+1}$, edge $e'_{i+1,j}$, the left-top boundary of $G_{i+1,j}$, and edge $e_{i,j}$.


We construct a proper good curve passing through $\Omega(g^2)\in \Omega(k^2)$ vertices of $G$. For simplicity of description, we construct a {\em closed} curve $\lambda$ passing through $\Omega(g^2)$ vertices of $G$ and such that, for each edge $e$ of $G$, either $\lambda$  contains $e$ or $\lambda$ has at most one point in common with $e$. Then $\lambda$ can be turned into a proper good curve by cutting off a piece of it in the interior of an internal face $f$ of $G$ and by changing the outer face of $G$ to $f$. 

Curve $\lambda$ passes through (at least) one vertex of each graph $G_{i,j}$ with $i$ and $j$ even, and with $4\leq i \leq g'$ and $2\leq j \leq g'$, where $g'$ is the largest integer divisible by $4$ and smaller than or equal to $g-2$; note that there are $\Omega(g^2)\in \Omega(k^2)$ such graphs $G_{i,j}$. Then Theorem~\ref{th:large-treewidth} follows from Theorem~\ref{th:topology}. Curve $\lambda$ is composed of several good curves, each one connecting two points in the interior of two reference edges for edges of $H$. Refer to Fig.~\ref{fig:cell}(b). In particular, each open curve is of one of the following types:

\begin{itemize}
	\item {\em Type A: Cell traversal curve.} A curve $\gamma$ connecting two points $p(\gamma)$ and $q(\gamma)$ in the interior of reference edges $e_{i,j}$ and $e_{i,j+1}$, or of reference edges $e'_{i,j}$ and $e'_{i+1,j}$. See, e.g., the part of $\lambda$ in the pink region in Fig.~\ref{fig:cell}(b).  
	
	\item {\em Type B: Cell turn curve.} A curve $\gamma$ connecting two points $p(\gamma)$ and $q(\gamma)$ in the interior of reference edges $e_{i,j}$ and $e'_{i,j}$, or of reference edges $e'_{i,j}$ and $e_{i,j+1}$, or of reference edges $e_{i,j+1}$ and $e'_{i+1,j}$, or of reference edges $e'_{i+1,j}$ and $e_{i,j}$. See, e.g., the part of $\lambda$ in the yellow region in Fig.~\ref{fig:cell}(b).  

	\item {\em Type C: Vertex getter curve.}  A curve $\gamma$ connecting two points $p(\gamma)$ and $q(\gamma)$ in the interior of reference edges $e'_{i,j-1}$ and $e'_{i+2,j}$ or of reference edges $e'_{i,j}$ and $e'_{i+2,j-1}$, and passing through a vertex of $G_{i+1,j}$. See, e.g., the part of $\lambda$ in the turquoise region in Fig.~\ref{fig:cell}(b).  
\end{itemize}
	
To each open curve $\gamma$ composing $\lambda$ we associate a distinct region $R(\gamma)$ of the plane, so that $\gamma$ lies in $R(\gamma)$. For curves $\gamma$ of Type A or B, the region $R(\gamma)$ is the unique cell delimited by the reference edges containing $p(\gamma)$ and $q(\gamma)$. For a curve $\gamma$ of Type C, the region $R(\gamma)$ consists of the interior of $G_{i+1,j}$ together with the four cells incident to the boundary of $G_{i+1,j}$. 

Any two regions associated to distinct open curves do not intersect, except along their boundaries. Further, for every region $R(\gamma)$ and for every edge $e$ of $G$, either $e$ is in $R(\gamma)$ or it has no intersection with the interior of $R(\gamma)$. Thus, in order to prove that $\lambda$ has at most one point in common with every edge of $G$, it suffices to show how to draw $\gamma$ so that it lies in the interior of $R(\gamma)$, except at points $p(\gamma)$ and $q(\gamma)$, and so that it has at most one common point with each edge in the interior of $R(\gamma)$. In order to describe how to draw $\gamma$, we distinguish the cases in which $\gamma$ is of Type A, B, or C. 

If $\gamma$ is of Type A or B (refer to Fig.~\ref{fig:curve-in-cell}(a)), draw the dual graph $D$ of $G$ so that each edge of $D$ only intersects its dual edge; restrict $D$ to the vertices and edges in the interior of $R(\gamma)$, obtaining a graph $D^*$; find a simple path $P$ in $D^*$ connecting the vertices $f_p$ and $f_q$ of $D^*$ incident to the reference edges to which $p(\gamma)$ and $q(\gamma)$ belong (note that $P$ exists since the region of the plane defined by each cell is connected and hence so is $D^*$); draw $\gamma$ as $P$ plus two curves connecting $f_p$ and $f_q$ with $p(\gamma)$ and $q(\gamma)$, respectively. Also, $\gamma$ intersects each edge of $G$ at most once, since $P$ does. Finally, $\gamma$ lies in the interior of $R(\gamma)$, except at points $p(\gamma)$ and $q(\gamma)$. Thus, $\gamma$ satisfies the required properties.

\begin{figure}[htb]
\begin{center}
\begin{tabular}{c c}
\mbox{\includegraphics[width=.35\textwidth]{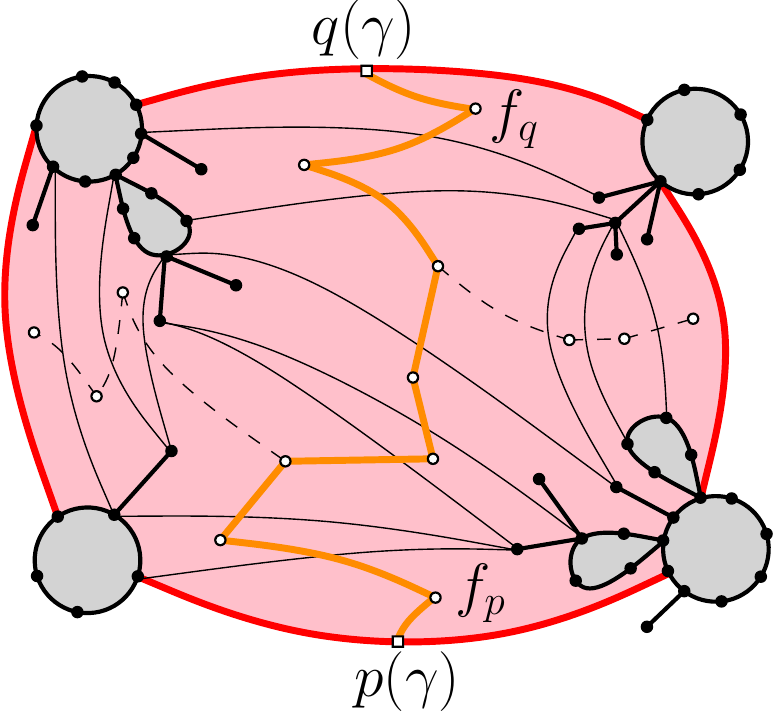}} \hspace{5mm} &
\mbox{\includegraphics[width=.5\textwidth]{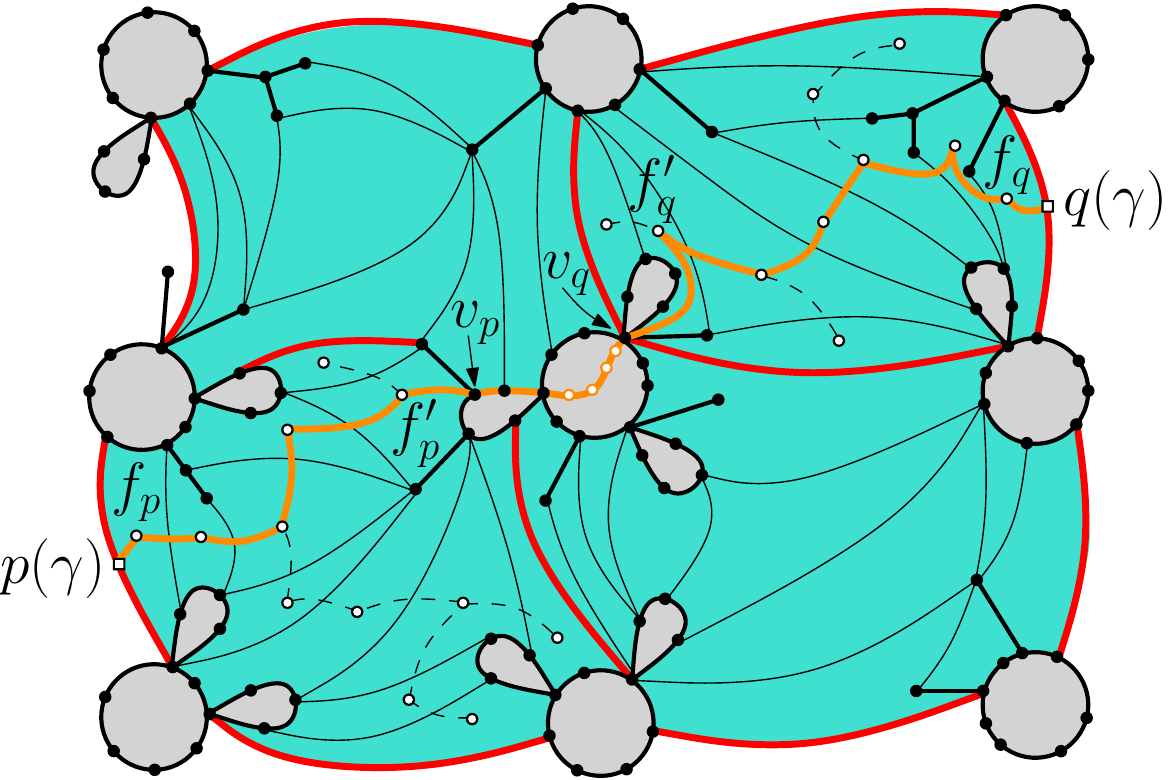}}\\
(a) \hspace{5mm} & (b)
\end{tabular}
\caption{(a) Drawing a curve $\gamma$ of Type A. Region $R(\gamma)$ is pink. Graph $D^*$ has vertices represented by white circles; the edges of $D^*$ in $P$ are thick orange lines, while the edges of $D^*$ not in $P$ are dashed black lines. (b) Drawing a curve $\gamma$ of Type C. Region $R(\gamma)$ is turquoise. Internal vertices of path $P$ in $G_{i+1,j}$ are black disks if they belong to the boundary of $G_{i+1,j}$, or orange circles if they are internal vertices of $G_{i+1,j}$.}
\label{fig:curve-in-cell}
\end{center}
\end{figure}

If $\gamma$ is of Type C (refer to Fig.~\ref{fig:curve-in-cell}(b)), assume that $\gamma$ connects two points $p(\gamma)$ and $q(\gamma)$ respectively belonging to the interior of $e'_{i,j-1}$ and $e'_{i+2,j}$; the case in which $p(\gamma)$ and $q(\gamma)$ respectively belong to the interior of $e'_{i,j}$ and $e'_{i+2,j-1}$ is analogous. Curve $\gamma$ is composed of three curves, namely: (1) a curve $\gamma_1$ that connects $p(\gamma)$ and a vertex $v_p$ on the left-bottom boundary of $G_{i+1,j}$, and that lies in the interior of $C_{i,j-1}$, except at $p(\gamma)$ and $v_p$; (2) a curve $\gamma_2$ that connects $v_p$ and a vertex $v_q$ on the right-top boundary of $G_{i+1,j}$, and that is an induced path in  $G_{i+1,j}$; and (3) a curve $\gamma_3$ that connects $v_q$ and $q(\gamma)$, and that lies in the interior of $C_{i+1,j}$, except at $v_q$ and $q(\gamma)$. Curve $\gamma_2$ might degenerate to be a single point $v_p=v_q$.

We start with $\gamma_2$. Consider a path $P$ in $G_{i+1,j}$ which is a shortest path connecting a vertex on the left-bottom boundary of $G_{i+1,j}$ and a vertex on the right-top boundary of $G_{i+1,j}$. Denote by $v_p$ and $v_q$ the end-vertices of $P$. Note that, possibly, $v_p=v_q$. Such a path $P$ always exists since $G_{i+1,j}$ is connected; also, $P$ has no internal vertex incident to the left-bottom boundary or to the right-top boundary of $G_{i+1,j}$, as otherwise there would exist a path shorter than $P$ between a vertex on the left-bottom boundary of $G_{i+1,j}$ and a vertex on the right-top boundary of $G_{i+1,j}$. Draw $\gamma_2$ as $P$. 

In order to draw $\gamma_1$ (curve $\gamma_3$ is drawn similarly), draw the dual graph $D$ of $G$ so that each edge of $D$ only intersects its dual edge; restrict $D$ to the vertices and edges in the interior of $C_{i,j-1}$, obtaining a graph $D^*$; find a shortest path $P_p$ in $D^*$ connecting the vertex $f_p$ of $D^*$ incident to the reference edge to which $p(\gamma)$ belongs and a vertex representing a face of $G$ incident to $v_p$. Denote by $f'_p$ the second end-vertex of such a path; draw $\gamma_1$ as $P$ plus two curves connecting $f_p$ and $f'_p$ with $p(\gamma)$ and $v_p$, respectively. 

Curve $\gamma$ has no intersections with the boundary of $R(\gamma)$ other than at $p(\gamma)$ and $q(\gamma)$. We now prove that $\gamma$ intersects each edge in $R(\gamma)$ at most once. First, $\gamma$ intersects each edge of $G_{i+1,j}$ at most once, since $\gamma_2$ is a shortest path in $G_{i+1,j}$ and since $\gamma_1$ and $\gamma_3$ have no intersections with the edges of $G_{i+1,j}$, except at $v_p$ and $v_q$. Second, $\gamma$ intersects each edge in $C_{i,j-1}$ at most once, since $P_p$ does, since $\gamma_1$ does not cross any edge incident to $v_p$ (given that $P_p$ is a shortest path between $f_p$ and any face incident to $v_p$), and since $\gamma_2$ and $\gamma_3$ do not intersect edges in $C_{i,j-1}$ other than at $v_p$ (given that $P$ does not contain any vertex incident to the left-bottom boundary of $G_{i+1,j}$ other than $v_p$); similarly, $\gamma$ intersects each edge in $C_{i+1,j}$ at most once. Third, $\gamma$ intersects each edge in $C_{i+1,j-1}$ at most once, namely at its possible end-vertex in $G_{i+1,j}$; similarly, $\gamma$ intersects each edge in $C_{i,j}$ at most once. Thus, $\gamma$ satisfies the required properties. 

This concludes the proof of Theorem~\ref{th:large-treewidth}.



\section{Implications for other graph drawing problems} \label{le:implications}

In this section, we present a number of corollaries of our results to other graph drawing problems. The following lemma is one of the key tools to establish these connections. For sake of completeness we explicitly state it here (in a more readily applicable form than the original, see \cite[Lemma 1]{DBLP:journals/dcg/BoseDHLMW09}).

\begin{lemma}\cite{DBLP:journals/dcg/BoseDHLMW09} \label{lem:help}
Let $G$ be a planar graph that has a \cfsl\ $\Gamma$ in which a (collinear) set $S\subseteq V(G)$ of vertices lie on the $x$-axis. Then, for an arbitrary assignment of $y$-coordinates to the vertices in $S$, there exists a \cfsl\ $\Gamma'$ of $G$ such that each vertex in $S$ has the same $x$-coordinate as in $\Gamma$ and has the assigned $y$-coordinate.
\end{lemma}

The above lemma immediately implies the following.

\begin{lemma}\cite{DBLP:journals/dcg/BoseDHLMW09} \label{lem:help2}
Let $G$ be a planar graph, $R\subseteq V(G)$ be a free collinear set, and $<_R$ be the total order associated with $R$. Consider any assignment of $x$- and $y$-coordinates to the vertices in $R$ such that the assigned $x$-coordinates are all distinct and the order by increasing $x$-coordinates of the vertices in $R$ is $<_R$ (or its reversal). Then there exists a \cfsl\ of $G$ such that each vertex in $R$ has the assigned $x$- and $y$-coordinates.
\end{lemma}

We first apply Lemma \ref{lem:help2} to obtain an optimal bound (up to a multiplicative constant) on the size of universal point subsets for planar graphs of treewidth at most three. 

\begin{corollary} \label{cor:sub}
Every set $P$ of at most  $\lceil\frac{n-3}{8}\rceil$ points in the plane is a universal point subset for all $n$-vertex plane graphs of treewidth at most three.
\end{corollary}

\begin{proof} 
If necessary, rotate the Cartesian axes so that no two points in $P$ have the same $x$-coordinate. By Theorems~\ref{th:3-trees} and~\ref{th:csvfcs} every $n$-vertex plane graph $G$ of treewidth at most three has a free collinear set $R$ of cardinality $|P|$. Let $<_R$ be the total order associated with $R$. Since no two points in $P$ have the same $x$-coordinate, there exists a bijective mapping $\delta:R\rightarrow P$ such that, for every two vertices $v,w\in R$, $v<_R w$ if and only if the $x$-coordinate of point $\delta(v)$ is smaller than the $x$-coordinate of point $\delta(w)$. Then by Lemma \ref{lem:help2} there exists a \cfsl\ of $G$ that respects mapping $\delta$.   
\end{proof}

It is implicit in \cite{DBLP:journals/dcg/BoseDHLMW09} and explicit in \cite{DBLP:conf/wg/RavskyV11} (in both cases using Lemmata \ref{lem:help} and \ref{lem:help2} above), that every straight-line drawing (possibly with crossings) of a planar graph $G$  can be untangled while keeping at least $\sqrt{x}$ vertices fixed, where $x$ is the size of a free collinear set of $G$. Together with Theorems~\ref{th:3-trees} and~\ref{th:csvfcs} this implies the following corollary.

\begin{corollary} \label{cor:untangle}
Any straight-line drawing (possibly with crossings) of an $n$-vertex planar graph of treewidth at most three can be untangled while keeping at least $\sqrt{\lceil (n-3)/8 \rceil}$ vertices fixed. 
\end{corollary}


We conclude this section with the application to column planar sets. Lemma \ref{lem:help} implies that every collinear set is a column planar set. That and our three main results imply our final corollary.

\begin{corollary} \label{cor:untangle2}
\begin{inparaenum}
\item[(a)] Triconnected cubic planar graphs have column planar sets of linear size.
\item[(b)] Planar graphs of treewidth at most three have column planar sets of linear size.
\item[(c)] Planar graphs of treewidth at least $k$ have column planar sets of size $\Omega(k^2)$.
\end{inparaenum}
\end{corollary}

\section{Conclusions}

In this paper we studied the problem of constructing planar straight-line graph drawings with many collinear vertices. It would be interesting to tighten the best known bounds (which are $\Omega(n^{0.5})$ and $O(n^{0.986})$) for the maximum number of vertices that can be made collinear in a \cfsl\ of any $n$-vertex planar graph. In particular, we ask: Is it true that, if a plane graph $G$ has a dual graph that contains a cycle with $m$ vertices, then $G$ has a \cfsl\ with $\Omega(m)$ collinear vertices? A positive answer to this question would improve the $\Omega(n^{0.5})$ lower bound to $\Omega(n^{0.694})$ (via the result in~\cite{j-lc3cg-86}). As noted in the Introduction, the ``converse'' is true for maximal plane graphs: If a maximal plane graph $G$ has a \cfsl\ with $x$ collinear vertices, then the dual graph $D$ of $G$ has a cycle with $\Omega(x)$ vertices.


We proved that every $n$-vertex triconnected cubic plane graph has a \cfsl\ with $\lceil \frac{n}{4} \rceil$ collinear vertices. It seems plausible that an $\Omega(n)$ lower bound holds true for every $n$-vertex subcubic plane graph. Recall from the Introduction that the linear lower bound does not extend to all bounded-degree planar graphs~\cite{DBLP:journals/dm/Owens81}, in fact, it does not extend already to all planar graphs of maximum degree $7$.

We proved that $n$-vertex plane graphs with threewidth at most three have \cfsl s with $\lceil \frac{n-3}{8} \rceil$ collinear vertices. Of our three results, this one has the widest applications to other graph drawing problems due to the fact that gives a free collinear set of size  $\lceil \frac{n-3}{8} \rceil$. In fact, we proved that every collinear set is a free collinear set in planar $3$-trees. This brings us to an open question already posed by Ravsky and Verbitsky~\cite{DBLP:conf/wg/RavskyV11}: is every collinear set a free collinear set, and if not, how close are the sizes of these two sets in a general planar graph?

Finally, we can prove that the maximum number of collinear vertices in any \cfsl\ of a plane $3$-tree $G$ can be computed in polynomial time (the statement extends to a {\em planar} $3$-tree by choosing the outer face in every possible way). Indeed, there are six (topologically distinct) ways in which a proper good curve $\lambda$ can ``cut'' the $3$-cycle $C$ delimiting the outer face of $G$: in three of them $\lambda$ passes through a vertex of $C$ and properly crosses the edge of $C$ not incident to that vertex, and in the other three $\lambda$ properly crosses two edges of $C$. This associates to $G$ six parameters, representing the maximum number of internal vertices of $G$ these six curves can pass through. Further, the six parameters for $G$ can be easily computed as a function of the same parameters for the plane $3$-trees children of $G$. This leads to a polynomial-time dynamic-programming algorithm to compute the six parameters and consequently the maximum number of collinear vertices in any \cfsl\ of $G$. By implementing this idea, we have observed the following fact: For every $m\leq 50$ and for every plane $3$-tree $G$ with $m$ internal vertices, there exists a \cfsl\ of $G$ with $\lceil \frac{m+2}{3} \rceil$ collinear internal vertices (and this bound is the best possible for all $m\leq 50$). It would be interesting to prove that this is the case for every $m\geq 1$.

\subsubsection*{Acknowledgments} The authors wish to thank Giuseppe Liotta, Sue Whitesides, and Stephen Wismath for stimulating discussions about the problems studied in this paper.

\bibliographystyle{abbrv}
\bibliography{bibliographyfinal}
\end{document}